\documentclass[11pt]{article}

\usepackage[a4paper, portrait, margin=2cm]{geometry}

\usepackage[colorlinks=true,linkcolor=blue,citecolor=ForestGreen]{hyperref}
\usepackage{algorithm}
\usepackage[noend]{algpseudocode}
\usepackage{url}
\usepackage{amsmath,amssymb,amsthm}
\usepackage{thmtools,thm-restate}
\usepackage[noabbrev,capitalise,nameinlink]{cleveref}
\usepackage{mathtools}
\usepackage{parskip}
\usepackage{xspace}
\usepackage{verbatim}
\usepackage{mathrsfs}
\usepackage[usenames,dvipsnames,svgnames,table]{xcolor}
\usepackage{pgf}
\usepackage[dvipsnames]{xcolor}
\usepackage{tabularx}
\usepackage{enumitem}
\usepackage{dsfont}

\theoremstyle{plain}
\newtheorem{theorem}{Theorem}[section]

\newtheorem{conjecture}[theorem]{Conjecture}
\newtheorem{corollary}[theorem]{Corollary}

\newtheorem{lemma}[theorem]{Lemma}
\newtheorem{fact}[theorem]{Fact}

\newtheorem{proposition}[theorem]{Proposition}

\newtheorem{observation}[theorem]{Observation}

\theoremstyle{definition}
\newtheorem{definition}[theorem]{Definition}
\newtheorem{question}[theorem]{Question}

\newtheorem*{claim*}{Claim}

\newcommand{\conv}{\mathsf{conv}}

\newcommand{\Eq}{\textsc{Eq}}
\newcommand{\Adj}{\textsc{Adj}}

\newcommand*\ie{i.\kern.1em e.\ }
\newcommand*\eg{e.\kern.1em g.\ }

\newcommand{\bc}[1]{\widetilde{#1}} % bipartite complement

\DeclareMathOperator\ch{\mathsf{ch}}

\DeclareMathOperator{\poly}{poly}
\DeclareMathOperator{\sign}{sign}
\DeclareMathOperator{\rank}{rank}

\renewcommand{\Pr}[1]{\bP \left[ #1 \right]} % Probability

\newcommand{\define}{:=}

\newcommand{\inn}[1]{\langle #1 \rangle}

\newcommand{\zo}{\{0,1\}}
\newcommand{\pmset}{\{\pm 1\}}

\newcommand{\cA}{\ensuremath{\mathcal{A}}}
\newcommand{\cB}{\ensuremath{\mathcal{B}}}

\newcommand{\cG}{\ensuremath{\mathcal{G}}}
\newcommand{\cH}{\ensuremath{\mathcal{H}}}

\newcommand{\cM}{\ensuremath{\mathcal{M}}}

\newcommand{\cP}{\ensuremath{\mathcal{P}}}
\newcommand{\cQ}{\ensuremath{\mathcal{Q}}}

\newcommand{\cX}{\ensuremath{\mathcal{X}}}

\newcommand{\bN}{\ensuremath{\mathbb{N}}}
\newcommand{\bP}{\ensuremath{\mathbb{P}}}
\newcommand{\bR}{\ensuremath{\mathbb{R}}}

\newcommand{\bZ}{\ensuremath{\mathbb{Z}}}

\newcommand\splitaftercomma[1]{%
  \begingroup
  \begingroup\lccode`~=`, \lowercase{\endgroup
    \edef~{\mathchar\the\mathcode`, \penalty0 \noexpand\hspace{0pt plus .25em}}%
  }\mathcode`,="8000 #1%
  \endgroup
}

\usepackage{tikz,pgfmath}
\usetikzlibrary{shapes,intersections,calc,backgrounds,scopes,decorations.pathreplacing}
\usepackage{subcaption}
\usepackage{float}
\usepackage{enumitem}
\usepackage{todonotes}
\usepackage{multirow}

\newtoggle{anonymous}
\togglefalse{anonymous}
\tikzstyle{b_vertex}=[circle,fill=black!100,text=white,inner sep=0.8mm,draw]
\tikzstyle{w_vertex}=[circle,fill=white!100,text=white,inner sep=0.8mm,draw]
\tikzstyle{w_s_vertex}=[circle,fill=white!100,text=white,inner sep=0.5mm,draw]
\tikzstyle{point}=[circle,fill=black,inner sep=0.1mm]
\tikzstyle{path_edge}=[thick]
\tikzstyle{eat edge} = [draw,BrickRed,line width=2pt]

\newcommand{\R}{\mathsf{R}}
\newcommand{\DEQ}{\mathsf{D}^\textsc{Eq}}

\usepackage{cite}

\title{Randomized Communication and Implicit Representations \\ for Matrices and Graphs of Small Sign-Rank}

\iftoggle{anonymous}{% ANONYMOUS
\author{Anonymous Authors}
}{%
\author{Nathaniel Harms\thanks{Some of this work was done while the author was at the University of
Waterloo, and visiting the University of Liverpool. Partly supported by an NSERC MSFSS, an NSERC
Postdoctoral Fellowship, and the Swiss State Secretariat for Education, Research and Innovation
(SERI) under contract number MB22.00026.}\\EPFL\\ \texttt{nathaniel.harms@epfl.ch}
\and
Viktor Zamaraev\\University of Liverpool\\ \texttt{viktor.zamaraev@liverpool.ac.uk}}
}

\newcommand{\Gyarfas}{Gy\'arf\'as }
\begin{document}
	
\maketitle

\begin{abstract}
We prove a characterization of the structural conditions on matrices of sign-rank 3 and unit disk
graphs (UDGs) which permit \emph{constant-cost} public-coin randomized communication protocols.
Therefore, under these conditions, these graphs also admit implicit representations.

The \emph{sign-rank} of a matrix $M \in \pmset^{N \times N}$ is the smallest rank of a matrix $R$
such that $M_{i,j} = \sign(R_{i,j})$ for all $i,j \in [N]$; equivalently, it is the smallest
dimension $d$ in which $M$ can be represented as a point-halfspace incidence matrix with halfspaces
through the origin, and it is essentially equivalent to the \emph{unbounded-error communication
complexity}. Matrices of sign-rank 3 can achieve the maximum possible \emph{bounded-error}
randomized communication complexity $\Theta(\log N)$, and meanwhile the existence of implicit
representations for graphs of bounded sign-rank (including UDGs, which have sign-rank 4) has been
open since at least 2003. We prove that matrices of sign-rank 3, and UDGs, have \emph{constant}
randomized communication complexity if and only if they do not encode arbitrarily large instances of
the \textsc{Greater-Than} communication problem, or, equivalently, if they do not contain large
half-graphs as semi-induced subgraphs. This also establishes the existence of implicit
representations for these graphs under the same conditions.

\end{abstract}

\thispagestyle{empty}
\setcounter{page}{0}
\newpage

% ToC spacing options, see tocloft package docs.
%\setlength{\cftbeforesecskip}{0.125em}

\thispagestyle{empty}
\setcounter{page}{0}
\newpage
{\small
\setcounter{tocdepth}{2} 
\tableofcontents
}
\thispagestyle{empty}
\setcounter{page}{0}
\newpage
\setcounter{page}{1}

\newpage
\section{Introduction}

Consider a sign matrix $M \in \pmset^{N \times N}$. In communication complexity, learning theory,
and graph theory, it is often useful to represent $M$ as a point-halfspace incidence matrix of the
following form. To each row $x \in [N]$, assign a point $p_x \in \bR^d \setminus \{0\}$, and to each
row $y \in [N]$ assign a unit vector $h_y \in \bR^d$, such that $M(x,y) = \sign(\inn{p_x, h_y})$. In
other words, $M(x,y) = 1$ if and only if the point $p_x$ belongs to the halfspace $H_y \define \{ p
\in \bR^d \;|\; \inn{p,h_y} \geq 0\}$ whose boundary hyperplane goes through the origin. It is
always possible to find such a representation, but, naturally, we wish to accomplish it in the
simplest way. Here are two common ways to measure the complexity of this representation:

\newcommand{\mar}{\mathsf{mar}}
%\begin{description}
%\item[Sign-rank.]
\textbf{Sign-rank.}
We might want to minimize the dimension $d$ of the representation. The minimum possible $d$
where $M$ admits such a representation is called the \emph{sign-rank} of $M$ and denoted
$\rank_\pm(M)$. It is equivalent to
the smallest rank $d$ of a matrix $R$ such that $M(x,y) = \sign(R(x,y))$ for all $x,y \in [N]$.
Thinking of the rows of $M$ as a fixed domain $\cX$, and the columns as a hypothesis class (\ie
subsets of $\cX$), a standard technique in learning theory is to transform the domain into points in
$\bR^d$, and the hypothesis class into halfspaces; $\rank_\pm(M)$ is the smallest dimension such that
this transformation is possible.  Since halfspaces through the origin in $\bR^d$ have VC dimension
$d$, sign-rank is lower bounded by the VC dimension of the hypothesis class. In
communication complexity, sign-rank is essentially equivalent to the \emph{unbounded-error}
communication complexity of $M$ \cite{PS86}, where the two players have access to \emph{private}
randomness and wish to succeed with probability strictly better than $1/2$. A set $\cM$ of
matrices has \emph{bounded sign-rank} if there exists a constant $d$ such that all matrices $M \in
\cM$ have sign-rank at most $d$. This is equivalent to having \emph{constant} unbounded-error
communication cost. In graph theory, finding \emph{implicit representations} (defined below) for
graphs whose adjacency matrices have bounded sign-rank is an open problem since at least 2003
\cite{Spin03}.

%\item[Margin.]
\textbf{Margin.}
We might want to maximize the \emph{margin} of the representation. For a fixed representation
$\{p_x\}_{x \in [N]}$ and $\{h_y\}_{y \in [N]}$, we define the \emph{margin} as $\min_{x,y}
\frac{|\inn{p_x,h_y}|}{\|p_x\| \cdot \|h_y\|}$. Write $\mar(M)$ for the maximum $m$ such that
there is a representation with margin $m$; the dimension of this representation is
irrelevant.  The complexity of various learning algorithms like SVM or perceptron can be bounded in
terms of the margin.  It is also known that $\mar(M)$ is functionally equivalent to the
\emph{two-way, public-coin randomized communication complexity} (\cref{prop:margin-discrepancy}). A
set $\cM$ of matrices has \emph{bounded margin} if there is some constant $m$ such that all $M \in
\cM$ have $\mar(M) \geq m$, and having bounded margin is equivalent to having \emph{constant}
public-coin randomized communication cost. Therefore, graphs whose adjacency matrices have bounded
margin admit implicit representations, due to the observation of \cite{Har20,HWZ22}.
%\end{description}

One of the main goals in communication complexity is to understand the power of randomness, and both
of the above measures of complexity capture a type of randomized communication.  A rapidly-growing
body of work on \emph{constant-cost communication}
\cite{Har20,HHH22dimfree,HWZ22,HHH22counter,EHK22,HHL22,HHPTZ22,ACHS23,CHHS23,HHM23} studies
the properties of matrices with bounded margin or bounded sign-rank, but the relationship
between these two measures is not well understood.  In one direction, it is believed that there
exist sets of matrices $\cM$ with bounded margin but unbounded sign-rank, but all known lower bounds
fail to prove this \cite{HHPTZ22} (although it was proven for \emph{partial} matrices
\cite{HHM23}). In this paper, we are interested in the other direction:

\begin{center}
\emph{For matrices $\cM$ of bounded sign-rank, under what conditions does $\cM$ also have bounded
margin?}\footnote{Note that a matrix having bounded sign-rank \emph{and} bounded margin does not mean that sign-rank
and margin are bounded \emph{simultaneously} by the same point-halfspace representation.}
\end{center}

\newcommand{\GT}{\mathsf{GT}}
It is known that \emph{some} conditions are required. Write $\R(M)$ for the two-way, public-coin
randomized communication cost of a matrix $M \in \pmset^{N \times N}$ (which we will refer to simply
as \emph{communication cost}) and $\R(\cM)$ for the communication cost of matrices $M \in \cM$
as a function of their size $N$ (see \cref{section:communication-complexity} for formal
definitions). The \textsc{Greater-Than} communication problem, defined by the matrices $\GT \in
\pmset^{N \times N}$ where $\GT_{i,j} = 1$ if and only if $i > j$, has sign-rank 2 but communication
cost\footnote{Standard notation in the literature uses $n$ as the number of bits in the input; we
use $N$ for the domain size, so $\Theta(\log\log N)$ corresponds to the more commonly-stated bound
$\Theta(\log n)$.} 
$\R(\GT) = \Theta(\log\log N)$ and therefore unbounded margin. When sign-rank increases to 3, 
matrices can achieve the maximum possible communication cost $\R(\cM) = \Theta(\log N)$
\cite{HHL22, ACHS23}, far exceeding the complexity of \textsc{Greater-Than}. However, one of our
main results is that, for sign-rank 3, \textsc{Greater-Than} is the only barrier to constant-cost
communication:

\begin{theorem}[Informal]
\label{thm:main-signrank-informal}
A set $\cM$ of matrices with sign-rank 3 has $\R(\cM) = O(1)$ (and therefore constant margin) if and
only if it does not contain arbitrarily large instances of \textsc{Greater-Than}.
\end{theorem}

We prove a similar theorem for the adjacency matrices of unit-disk graphs (UDGs), which have
sign-rank 4, and these results establish the existence of implicit representations when the
condition on the \textsc{Greater-Than} instances is satisfied.  We also exhibit a
fundamental gap between sign-rank 4 and 5 which shows that the ``type'' of randomness used in
our communication protocols cannot succeed in sign-rank 5 and above.
\cref{thm:main-signrank-informal} is a consequence of more general results whose motivation and
applications we elaborate upon below.

\subsection{Constant-Cost Communication and Implicit Graph Representations}

The study of constant-cost randomized communication was initiated independently in
\cite{HHH22dimfree,HWZ22}. One motivation of \cite{HWZ22} was that constant-cost communication
is a special case of a well-studied open problem in structural graph theory and distributed
computing, which asks to characterize the hereditary graph classes $\cG$ that admit \emph{implicit
representations} (see \eg \cite{KNR92,Spin03,ACLZ15,HH22,AAALZ23,Chan23}).

\newcommand{\enc}{\mathsf{enc}}
\paragraph*{Implicit representations.}
A \emph{class} of graphs is a set $\cG$ of (labeled) graphs that is closed under isomorphism. It is
\emph{hereditary} if it is closed under taking induced subgraphs. A hereditary class $\cG$ admits an
\emph{implicit representation} if there exists a decoder
$D : \zo^* \times \zo^* \to \zo$ such that, for every $N$-vertex graph $G \in \cG$, each vertex $v$
of $G$ can be assigned an encoding $\enc(v)$ of $O(\log N)$ bits, where $D(\enc(u),
\enc(v))$ outputs the adjacency of vertices $u,v$; the decoder $D$ depends on the class $\cG$ but
\emph{not} the specific graph $G$. Implicit representations were introduced in \cite{KNR92}, who
observed that they are equivalent to a graph $U$ of size $\poly(N)$, called a \emph{universal
graph}, that contains every $N$-vertex graph $G \in \cG$ as an induced subgraph. Since a graph of
size $\poly(N)$ has at most $2^{O(N \log N)}$ $N$-vertex induced subgraphs, a necessary condition
for the existence of implicit representations is that $\cG$ contains at most $2^{O(N \log N)}$
$N$-vertex graphs, in which case $\cG$ is said to have \emph{factorial speed}.

The communication problem defined by any matrix $M \in \pmset^{N \times N}$ is equivalent
to the problem of deciding adjacency in the (bipartite) graph whose adjacency matrix is $M$, where
each player is given a vertex. Building on \cite{Har20}, \cite{HWZ22} observed that constant-cost
communication problems $\cM$ are \emph{equivalent} to hereditary graph classes $\cG$ that admit an
\emph{adjacency sketch}, which is a randomized version of an implicit representation, where the
encodings $\enc(v)$ are assigned by a randomized algorithm and have \emph{constant size}
(independent of the number of vertices), in such a way that
\[
  \forall u,v : \qquad \Pr{ D(\enc(u), \enc(v)) \text{ correctly outputs adjacency of } u,v } \geq
2/3 \,.
\]
Adjacency sketches for trees also appeared earlier in \cite{FK09}. As noted in \cite{Har20,HWZ22},
adjacency sketches can be derandomized (see \cref{section:communication-to-labeling}) to obtain
implicit representations, making constant-cost randomized communication protocols a stronger type of
implicit representation. 

\paragraph*{Unit disk graphs.}
This again motivates our focus on sign-rank. Graphs whose adjacency matrices have bounded sign-rank
are among the most important types of graphs for which implicit representations are not known to
exist in general: to obtain implicit representations for geometric intersection graphs (more precisely,
semi-algebraic graphs), it suffices to study graphs of bounded sign-rank (see \eg
\cite{Fit19}). Any class of bounded sign-rank satisfies the necessary condition of \emph{factorial
speed} \cite{Spin03}, which was conjectured to be sufficient in \cite{KNR92,Spin03}.  Until this
conjecture was refuted in \cite{HH22} by a non-constructive argument, classes of bounded sign-rank
were considered promising candidates for a counterexample \cite{Spin03}. The best
known implicit representations for classes of bounded sign-rank in general use $O(N^{1-\epsilon})$
bits per vertex where $\epsilon > 0$ is a constant \cite{Fit19,Alo22}.

A canonical example is the unit disk graphs (UDGs).  UDGs admit an ``implicit representation'' in
the sense that each vertex may be encoded with the coordinates of its disk in $\bR^2$. However, this
encoding requires exponentially-many bits \cite{MM13}, and it is a central open problem whether this
difficulty can be sidestepped to obtain encodings of size $O(\log N)$; our understanding is that
this is not widely believed to be possible. In this paper, we resolve the \emph{randomized} version
of the question by giving a complete characterization of the UDGs which admit constant-size
adjacency sketches. To state this result, we require the notion of \emph{stability} (see \eg
\cite{HWZ22,POS22}).

\paragraph*{Stability.} The \emph{chain-index} $\ch(G)$ of a graph $G$ is the largest $k$
such that there exist disjoint sets of vertices $\{a_1, \dotsc, a_k\}$ and $\{b_1, \dotsc, b_k\}$
where, for any $i < j$, $a_i, b_j$ are adjacent but $b_i, a_j$ are not.  In the
terminology of \cite{GPT22}, a graph class $\cG$ is \emph{graph-theoretically stable} if there is a
constant $k$ such that $\ch(G) \leq k$ for all $G \in \cG$; we will say simply
\emph{stable}\footnote{We use \emph{stable} in this paper but we note that the disambiguation
\emph{graph-theoretically stable} in \cite{GPT22} is necessary to avoid confusion with stability in
the literature on model theory.}.
The chain-index is essentially\footnote{Not exactly: we have no restriction on the adjacency between
$a_i, b_i$, which
helps the analysis but is not qualitatively important.} the largest instance of the
\textsc{Greater-Than} communication problem that appears in $G$, and therefore a class that is
\emph{not} stable must have non-constant communication cost (see \cite{HWZ22,HHH22counter} for more
on the stability condition in communication).
For a graph class $\cG$, write
$\R(\cG)$ for the function $N \mapsto \max_G \R(\Adj_G)$ where $G$ ranges over the $N$-vertex graphs
in $\cG$ and $\Adj_G$ is the adjacency matrix of $G$ (if $\cG$ is a class of bipartite graphs, we
take the bipartite adjacency matrix). Stability is necessary for $\R(\cG) = O(1)$; for UDGs and
graphs of sign-rank 3, we show it is also \emph{sufficient}:

\begin{theorem}
\label{thm:intro-main-udg}
Let $\cG$ be either a subclass of UDGs, or a class of sign-rank at most 3. Then $\R(\cG) = O(1)$ if
and only if $\cG$ is stable. As a consequence, stable subclasses of UDGs and graphs of sign-rank 3
admit implicit representations.
\end{theorem}

\subsection{Results and Techniques}

\cref{thm:intro-main-udg} follows from a more general result that has other implications for
implicit graph representations and which unifies and generalizes a number of previous results. We
also complement it with an impossibility result that rules out using the type of randomized
techniques in this paper to prove similar results in sign-rank 5 and above. Let us now explain these
results in more detail and give a brief summary of the techniques.

\paragraph*{Constant-cost reductions.}
We require the notion of \emph{constant-cost reductions} and the \textsc{Equality} oracle.  The
\textsc{Equality} communication problem is the standard example of the power of (public-coin)
randomized communication. Two players are given inputs $x, y \in [N]$, respectively, and they must
decide if $x = y$. By random hashing, this can be done with success probability $3/4$ using only 2
bits of communication. The success probability can be improved to any arbitrary constant
by increasing the number of bits by a constant factor.

One way to design a constant-cost communication protocol is to design a \emph{deterministic}
communication protocol with constant cost, which has access to an oracle that computes
\textsc{Equality}. This means that the two players can, at any time, supply the oracle with
arbitrary values $a,b$ and receive, at unit cost, the answer to the query ``$a = b$?''
The power of the \textsc{Equality} oracle has been studied in several works
\cite{BFS86,CLV19,HHH22dimfree,HWZ22,EHK22,AY22,CHHS23,PSS23}.
One may think of these protocols as the ones that can be implemented using standard practical hash
functions like SHA256.  Constant-cost protocols of this form are examples of \emph{constant-cost
reductions}, a type of reduction that is natural for both constant-cost communication complexity and
implicit graph representations; we formally define constant-cost reductions in general in
\cref{section:reductions}. Along with the algorithmic definition of reductions to \textsc{Equality},
there is an equivalent structural definition (see \eg \cite{HWZ22,AY22}): if a graph class $\cG$
admits a constant-cost protocol for computing adjacency in graphs $G \in \cG$, using
\textsc{Equality} oracles, then there exists a constant $t$ such that the adjacency matrix $\Adj_G$
of every graph $G \in \cG$ (or bipartite adjacency matrix, if $\cG$ is a class of bipartite graphs)
can be written as
\[
  \forall x,y :\qquad \Adj_G(x,y) = f(Q_1(x,y), Q_2(x,y), \dotsc, Q_t(x,y)) \,,
\]
where $f : \zo^t \to \zo$ and each $Q_i$ is the bipartite adjacency matrix of a bipartite
\emph{equivalence graph} (disjoint union of bicliques). We write $\DEQ( M )$ for the minimum cost of
a 2-way deterministic protocol with \textsc{Equality} oracles.  For computing adjacency in
\emph{monotone} graph classes (closed under edge \& vertex deletions), all constant-cost randomized
protocols can be put in this form \cite{EHK22}, but in general they cannot
\cite{HHH22dimfree,HWZ22}.  \cite{HHPTZ22} showed that $\DEQ(\cG) = O(1)$ implies that $\cG$ has
bounded sign-rank; our results explore the converse.

\paragraph*{Forbidden cycles and subdivided stars.}
Our \cref{thm:intro-main-udg} is a consequence of a more general result, \cref{thm:intro-main}
below, which also makes some progress towards characterizing the \emph{finitely-defined} bipartite
graph classes for which constant-cost communication and implicit representations are possible.  For
any set $\cH$ of bipartite graphs, a class $\cG$ of bipartite graphs is \emph{$\cH$-free} if no
graph $G \in \cG$ contains any $H \in \cH$ as an induced subgraph. Every hereditary class of
bipartite graphs is $\cH$-free for some unique but possibly infinite set $\cH$. For fixed $\cH$,
write $\cG_\cH$ for the $\cH$-free bipartite graphs.
For a bipartite graph $G=(U,W,E)$ with a fixed bipartition, we write $\bc{G}$ for the bipartite complement of $G$, i.e. $\bc{G}=(U,W,(U \times W) \setminus E)$.

The condition that $\cG$ is \emph{stable} is equivalent to the condition that it is $H_k$-free for
some constant $k$, where $H_k$ denotes the \emph{half-graph} (see \cite{HWZ22}), so $\R(\cG_\cH) =
O(1)$ requires that $\cH$ contain some half-graph $H_k$. When $\cH$ is \emph{finite}, it is also
necessary that $\cH$ contain both a tree and the bipartite complement of a tree,
otherwise the number of graphs in $\cG_\cH$ is too large \cite{LZ17}. In the case $|\cH| = 2$, it is
therefore necessary for $\R(\cG_\cH) = O(1)$ that $\cH = \{H_k, T\}$ where $T$ and its bipartite
complement $\bc{T}$ are both trees; it was proved in \cite{HWZ22} that this is also
\emph{sufficient}. We believe these conditions remain sufficient for larger (but still
finite) $\cH$, \ie $\R(\cG_\cH) = O(1)$ whenever $\cH = \{ H_k, T_1, \bc{T_2}\}$ for some trees $T_1$ and
$T_2$.  When $T_1$ and $T_2$ are \emph{subdivided stars}, our result confirms this.

\begin{definition}[Subdivided Star]
\label{def:subdivided-star}
For $s,t \in \bN$, we write $S_{s,t}$ for the
\emph{subdivided star}, which is obtained by taking the star graph with $s$ leaves and subdividing
each edge $t-1$ times. 
\end{definition}

As usual, we denote by $C_t$ the cycle on $t$ vertices.
%For a set of bipartite graphs $\cH$, we say that a class $\cG$ of bipartite graphs is \emph{$\cH$-free} if every graph $G \in \cG$ contains no graph from $\cH$ as an induced subgraph. 
Our main technical result is:

\begin{restatable}{theorem}{thmintromain}
\label{thm:intro-main}
Let $\cG$ be a stable class of bipartite graphs that satisfies either of these conditions:
\begin{enumerate}
\item There exist constants $s, t$ such that $\cG$ is $(S_{s,t}, \bc{S}_{s, t})$-free;
or
\item There exists a constant $t$ such that $\cG$ is $\{ C_{t'}, \bc{C}_{t'} ~|~ t' \geq t \text{ and } t' \text{ is even} \}$-free.
\end{enumerate}
Then $\DEQ(\cG) = O(1)$.
\end{restatable}

We use \cref{thm:intro-main} to prove \cref{thm:intro-main-udg} by decomposing UDGs or graphs of
sign-rank 3 into bipartite graphs that are both $(S_{3, 3}, \bc{S}_{3, 3})$-free and $\{ C_{t},
\bc{C}_{t} ~|~ t \geq  10 \text{ and } t \text{ is even}\}$-free (which, to clarify, is stronger
than necessary to apply the theorem). We remark that the implicit representation implied by
\cref{thm:intro-main} can be efficiently computed, meaning that the labels can be constructed in
time $\poly(N)$ and decoded in time $\poly\log N$. This efficiency is inherited by the implicit
representations of UDGs and graphs of sign-rank 3, provided that the encoder is given the geometric
representation of the input graph.

\cref{thm:intro-main} is much more general, and also allows us to
recover several prior results. Analogs of \cref{thm:intro-main-udg,thm:intro-main} for the classes
of permutation graphs, interval graphs, and $P_7$-free and $S_{1,2,3}$-free bipartite graphs were
proved in \cite{HWZ22}. All of these results, which in \cite{HWZ22} each required different proof
strategies, follow as corollaries of \cref{thm:intro-main,thm:intro-main-udg}. Likewise,
\cite{AAALZ23} showed the existence of implicit representations for stable, chordal bipartite
graphs, which is also implied by \cref{thm:intro-main}.

\paragraph*{Higher sign-ranks and weakly-sparse graphs.}
To advance beyond sign-rank 3, it is helpful to compare the \emph{stability} condition with the
stronger \emph{weakly-sparse} condition. A class of graphs $\cG$ is \emph{weakly-sparse} if there is
a constant $t$ such that no graph $G \in \cG$ contains $K_{t,t}$ as a subgraph. Any weakly-sparse
class is also stable. It is known and not difficult to prove that any weakly-sparse subclass of
UDGs has bounded degeneracy, and therefore the analog of \cref{thm:intro-main-udg} for weakly-sparse
UDGs is trivial (because $\DEQ(\cG) = O(1)$ for any $\cG$ of bounded degeneracy). For weakly-sparse
graph classes, we present a proof in \cref{section:hierarchy} that reductions to \textsc{Equality}
are \emph{equivalent} to bounded degeneracy:

\begin{restatable}{theorem}{thmeqlowerbound}
\label{thm:eq-lower-bound}
Let $\cG$ be a hereditary class of bipartite graphs that is weakly-sparse. Then $\DEQ(\cG) = O(1)$
if and only if $\cG$ has bounded degeneracy.
\end{restatable}

In \cite{CHHS23}, it is conjectured that the point-line incidence graphs $\mathcal{PL}$ satisfy
$\R(\mathcal{PL}) = \omega(1)$. \cref{thm:eq-lower-bound} shows the weaker result $\DEQ(\cG) =
\omega(1)$, because point-line incidences are $K_{2,2}$-free and have unbounded degeneracy. They
also have sign-rank at most 6, which means that the \textsc{Equality} oracle does not suffice to
extend \cref{thm:intro-main} to sign-rank 6 and above, even if the stability condition
is replaced with the much stronger \emph{weakly-sparse} condition. Combining known results in the
literature, we also give in \cref{section:hierarchy} an example ($K_{2,2}$-free point-box incidence
graphs) with sign-rank 5 that is $K_{2,2}$-free but has unbounded degeneracy, showing in fact that
the \textsc{Equality} oracle does not suffice to extend \cref{thm:intro-main} to sign-rank 5.
It may be the case that reductions to \textsc{Equality} are the only type of constant-cost
communication possible for matrices of bounded sign-rank, see \cref{question:constant-cost}.
We summarize the known results for low sign-ranks in \cref{table:hierarchy}.

\newcommand{\resultcell}{\cellcolor{YellowGreen}}
\begin{table}
\centering
\begin{tabular}{c|c|c|c|c}
 & \multicolumn{4}{c|}{Sign-rank} \\
\cline{2-5}
 & 3 & 4 & 5 & 6 \\
\hline
\multirow{2}{*}{Weakly-sparse} & $\DEQ = O(1)$ & $\DEQ = O(1)$ & \resultcell $\DEQ =
\omega(1)$ & \resultcell $\DEQ = \omega(1)$ \\
\cline{2-5}
 & $\R = O(1)$  & $\R = O(1)$ & $\R=?$ & $\R = $ conjectured $\omega(1)$ \\
\hline
\multirow{2}{*}{Stable} & \resultcell $\DEQ = O(1)$ & \resultcell $\DEQ = O(1)$ for UDGs & \resultcell $\DEQ = \omega(1)$ & \resultcell $\DEQ = \omega(1)$ \\
\cline{2-5}
 & \resultcell $\R = O(1)$ & \resultcell $\R = O(1)$ for UDGs & $\R=?$ & $\R = $ conjectured $\omega(1)$ \\
\hline
\end{tabular}
\caption{The landscape of communication complexity for small sign-ranks; see
\cref{section:hierarchy}. Cells in green are results
of this paper. The conjectures for sign-rank 6 are due to \cite{CHHS23}. The upper bounds for
weakly-sparse and sign-rank $\leq 4$ follow from \cite{CH23} (and we also give an alternate proof of a
weaker statement in \cref{section:caratheodory}).}
\label{table:hierarchy}
\end{table}

\paragraph*{Proof overview.}
We briefly summarize the proofs of \cref{thm:intro-main,thm:intro-main-udg}. Although UDGs and
graphs of sign-rank 3 do not satisfy the conditions of \cref{thm:intro-main}, we prove that two
parties with access to an \textsc{Equality} oracle can agree on a graph decomposition into pieces
that avoid \emph{edge-asteroid triple} structures (used in \eg \cite{FHH99,STU10,AZ18}), which
guarantees that these pieces satisfy the conditions of \cref{thm:intro-main}.

Our main tool to prove \cref{thm:intro-main} is the \emph{\Gyarfas decomposition}, which we take
from \cite{POS22}. The \Gyarfas decomposition partitions a bipartite graph into bags of vertices
with a tree-like structure on the bags that controls the edges between the bags. In particular,
every root-to-leaf path on the bags induces a path in the original graph. For this reason, the
\Gyarfas method has previously been used (as in \cite{POS22}) to analyze $P_t$-free graphs, \ie
graphs which forbid long induced paths, where the depth of the decomposition is constant.

However, in our case, the depth of the decomposition is unbounded. Instead, we show that, under the
conditions of \cref{thm:intro-main}, each bag has edges to only a bounded number of its ancestors.
Using this guarantee, we show that a communication protocol on input vertices $x,y$ may use the
\textsc{Equality} oracle to either determine the adjacency, or agree on a subset of bags that
contains $x$ and $y$. The protocol may then recurse on these bags, sometimes switching to the
bipartite complement of the graph when it does so (this is why \eg we require both $S_{s, t}$
\emph{and} $\bc{S}_{s, t}$ to be forbidden). Due to arguments of \cite{POS22}, this recursion will
reduce the chain-index of the graph and is therefore guaranteed to terminate after a constant number
of iterations.

\subsection{Discussion and Open Problems}

\paragraph*{Communication complexity.}
An intriguing possibility arises from this work, in conjunction with other
recent work on bounded sign-rank. Adapting (or abusing) some notation of \cite{BFS86}, write
$\mathsf{UPP}[1]$ for the set of communication problems with bounded sign-rank (\ie constant
\emph{unbounded-error communication cost} \cite{PS86}), write $\mathsf{BPP}[1]$ for the set of
communication problems with constant public-coin randomized communication cost, and write $\DEQ[1]$
for the set of communication problems with a constant-cost reduction to \textsc{Equality}. With
these definitions of communication complexity classes, we can ask:

\begin{question}
\label{question:constant-cost}
Is it the case that $\DEQ[1] = \mathsf{UPP}[1] \cap \mathsf{BPP}[1]$?
\end{question}

A positive answer to this question would ``explain'' all of the known results and conjectures
relating these classes. It is proved in \cite{CHHS23} that $\DEQ[1] \subseteq \mathsf{UPP}[1] \cap
\mathsf{BPP}[1]$. In the other direction, there are communication problems in $\mathsf{BPP}[1]$ that
do not belong to $\DEQ[1]$, which was proved independently in \cite{HHH22dimfree} and \cite{HWZ22},
but the example in both cases, the \textsc{$1$-Hamming Distance} problem (adjacency in the
hypercube), is believed not to belong to $\mathsf{UPP}[1]$ \cite{HHPTZ22}, which is implied by a
positive answer to \cref{question:constant-cost}. In \cref{section:hierarchy}, we give two explicit examples
($K_{2,2}$-free point-box incidences, and point-line incidences) in $\mathsf{UPP}[1]$ that do not
belong to $\DEQ[1]$, which could possibly provide a negative answer to \cref{question:constant-cost}
if they belong to $\mathsf{BPP}[1]$, but point-line incidences are conjectured not to belong to
$\mathsf{BPP}[1]$ in \cite{CHHS23}.  On the other hand, a negative answer to
\cref{question:constant-cost} seems to require a substantially different type of randomized protocol
than the ones which have so far been discovered\footnote{By this we mean that it seems unlikely to
us that a negative answer to the question would be achieved by a reduction to any currently-known
constant-cost problem, most of which can be found in \cite{HWZ22}.}, and would therefore be very
interesting.

\paragraph*{Implicit representations.}
An obvious question is whether the stability condition in our positive result for implicit
representations can be dropped. This cannot be accomplished by reductions to \textsc{Equality}, for
which stability is necessary. We have shown that the \textsc{Greater-Than} problem is the only
barrier to constant-cost communication, so one idea for generalizing our result is to allow the more
powerful \textsc{Greater-Than} oracles in the communication protocol. Constant-cost reductions to
\textsc{Greater-Than} are equally good for the purpose of finding implicit representations (we may
think of some standard implicit representations, like for interval graphs \cite{KNR92} and point-box
incidences \cite{Spin03}, as protocols of this form). But this cannot succeed: a
constant-cost reduction to \textsc{Greater-Than} for graphs $\cG$ of sign-rank 3 would imply
$\R(\cG) = \Theta(\log\log N)$ which contradicts the known bound of $\Theta(\log N)$
\cite{HHL22,ACHS23}. This answers an open
question asked in independent and concurrent work \cite{Chan23} whether (in our terminology)
reductions to \textsc{Greater-Than} suffice to obtain implicit representations for geometric
intersection graphs with small sign-rank realized by integer coordinates\footnote{The bounds in
\cite{HHL22,ACHS23} hold for constructions with integer coordinates.}. 
This at least demonstrates
that communication complexity lower bounds can be used against certain natural types of implicit
representation, although it remains open how to prove \emph{any} explicit, non-trivial lower bounds
for implicit representations.

\section{Preliminaries}

Let us define some notation and formalize the notions we have discussed in the introduction. We
intend this paper to be accessible to readers in graph theory or communication complexity who may
not have a background in both, so we make an attempt to make the terminology explicit. We will also
define a general notion of \emph{constant-cost reductions} which has not yet appeared explicitly
in the literature.

\subsection{Notation}

For a matrix $M \in \bR^{X \times Y}$, row $x \in X$, and column $y \in Y$, we will write either
$M_{x,y}$ or $M(x,y)$ for the entry at $x$ and $y$.

For a graph $G$, we write $\overline G$ for the complement of $G$.  For a bipartite graph $G =
(X,Y,E)$ with a fixed bipartition, write $\bc{G}$ for the bipartite complement, which has edge
$xy$ if and only if $xy$ is not an edge of $G$. The adjacency matrix of a graph $G = (V,E)$ is the
matrix $\Adj_G \in \pmset^{V \times V}$ with $\Adj_G(x,y) = 1$ if and only if $xy \in E$. For a
bipartite graph $G = (X,Y,E)$ with a fixed bipartition, the \emph{bipartite} adjacency matrix is the
matrix $\Adj_G \in \pmset^{X \times Y}$ with $\Adj_G(x,y) = 1$ iff $xy \in E$, where we note that
the rows are indexed by $X$ instead of the full set of vertices $X \cup Y$ (and similar for the
columns).

For a graph $G$ and disjoint sets $X,Y \subseteq V(G)$, we will write $G[X,Y]$ for the
\emph{semi-induced} bipartite subgraph, which is the bipartite graph $G[X,Y] = (X,Y,E)$ defined by
putting an edge between $x \in X$ and $y \in Y$ if and only if $xy$ are adjacent in $G$. (In
particular, any edges within $X$ or $Y$ in $G$ are not present in $G[X,Y]$.)

\subsection{Sign-Rank}

For a matrix $M \in \pmset^{N \times N}$, the \emph{sign-rank} of $M$ is denoted $\rank_\pm(M)$ and
it is the minimum $d \in \bN$ such that there exists a matrix $R \in \bR^{N \times N}$ of rank $d$
with $M = \sign(R)$, where $\sign(R) \in \pmset^{N \times N}$ is the matrix with entries
\[
  \forall i,j \in [N] : \sign(R)_{i,j} = \sign(R_{i,j}) \,.
\]
Equivalently, $\rank_\pm(M)$ is the minimum $d$ such that each row $i \in [N]$ may be associated
with a unit vector $p_i \in \bR^d$ (which we think of as a point) and each column $j \in [N]$ may be
associated with a unit vector $h_j \in \bR^d$ (which we think of as the normal vector for a
halfspace), such that $M_{i,j} = \sign(\inn{p_i, h_j})$. In this way, the sign-rank of $M$ is
equivalent to the minimum dimension $d$ such that $M$ is the incidence matrix between a set of
points $X$ and a set of halfspaces $Y$, where the hyperplane boundaries of the halfspaces contain
the origin.

We require a notion of sign-rank for graphs, which we will define separately for bipartite graphs
with a fixed bipartition, and for general graphs.  For a bipartite graph $G = (X,Y,E)$ with a fixed
bipartition, its sign-rank $\rank_\pm(G)$ is defined as the sign-rank $\rank_\pm(\Adj_G)$ of its
bipartite adjacency matrix $\Adj_G \in \pmset^{X \times Y}$. For a general graph $G = (V,E)$, we
define its \emph{partial} adjacency matrix $A \in \{\pm 1, \star\}^{V \times V}$ to be
\[
  \Adj_G^*(x,y) \define \begin{cases}
  \star &\text{ if } x = y \\
  1 &\text{ if } xy \in E \\
  -1 &\text{ otherwise.}
  \end{cases}
\]
We then define the sign-rank $\rank_\pm(G)$ as the minimum rank of a matrix $R$ such that
\[
  \forall i \neq j : \sign(R_{i,j}) = \Adj_G^*(i,j) \,.
\]
Specifically, we do not make any requirement on the diagonal entries.

\subsection{Communication Complexity and Margin}
\label{section:communication-complexity}

For a matrix $M \in \pmset^{N \times N}$, we will write $\R(M)$ for the \emph{public-coin randomized
communication complexity} of $M$, with success probability $2/3$. In this model, Alice receives a
row $x \in [N]$ and Bob receives a column $y \in [N]$ and they must output $M(x,y)$. They are given
shared access to a string of random bits, and they take turns sending messages that depend on
their respective inputs and the random string. They must output the correct answer with probability
at least $2/3$ over the random string, and the complexity of a protocol is the total number of bits
communicated between the players on the worst-case inputs $x,y$. $\R(M)$ is the minimum complexity
of any such protocol computing $M$. 
See \cite{KN96,RY20}.

The standard notion of a (total, Boolean-valued) \emph{communication problem} is a sequence $\cP =
(P_N)_{N \in \bN}$ of matrices, where $P_N \in \pmset^{N \times N}$, and the complexity of the
problem, denoted $\R(\cP)$, is the function $N \mapsto \R(P_N)$. However, we are interested in the
complexity of \emph{classes} of matrices (specifically adjacency matrices of graphs belonging to
some graph class), not merely sequences of matrices, where there is a variety of $N \times N$
matrices instead of just one. So we define communication problems more generally, as in
\eg \cite{HWZ22,EHK22}.

\begin{definition}[Communication Problem]
A \emph{communication problem} is a set $\cP = \bigcup_{N \in \bN} \cP_N$ of Boolean matrices, where
$\cP_N$ is a finite set of matrices in $\pmset^{N \times N}$. We then define the communication complexity
$\R(\cP)$ as the function
\[
  N \mapsto \max_{P \in \cP_N} \R(P) \,.
\]
For a class $\cG$ of graphs, we write $\Adj_\cG$ for the communication problem that is the set of
adjacency matrices of graphs in $\cG$. If $\cG$ is a class of bipartite graphs, we take the
bipartite adjacency matrices. We abuse notation and write $\R(\cG) = \R(\Adj_\cG)$, so that
$\R(\cG)$ is the function
\[
  N \mapsto \max \{ R(\Adj_G) \;|\; G \in \cG \text{ has $N$ vertices } \} \,.
\]
\end{definition}
Communication complexity is always upper bounded by the number of bits $n$ in the input, or in our
notation, by $\lceil \log N \rceil$.  We are interested in determining which communication problems
have \emph{constant cost}, which means that there exists a constant $c$ such that $\R(M) \leq c$ for
all $M \in \cP$. One way to rule out a constant-cost protocol for a problem $\cP$ is if the
\textsc{Greater-Than} communication problem appears as a subproblem of $\cP$. Formally, this is
captured by the stability condition (see \cite{HWZ22}):

\begin{proposition}
\label{prop:gt-lower-bound}
Let $\cG$ be any graph class which is not stable. Then $\R(\cG) = \omega(1)$.
\end{proposition}

As mentioned in the introduction, having constant communication cost is equivalent to having
constant \emph{margin}, due to the following inequality, which follows from results of \cite{LS09}:

\begin{proposition}
\label{prop:margin-discrepancy}
Let $M \in \pmset^{N \times N}$. Then
\[
  \Omega\left(\log\frac{1}{\mar(M)}\right) \leq \R(M) \leq O\left(\frac{1}{\mar(M)^2}\right) \,.
\]
\end{proposition}

\subsection{Constant-Cost Communication Reductions and Equality}
\label{section:reductions}

One way to obtain constant-cost protocols is by reduction to the \textsc{Equality} problem, for
which we require the definitions of the \textsc{Equality} problem and a notion of reduction.

\begin{definition}[The \textsc{Equality} Communication Problem]
The \textsc{Equality} communication problem is the set $\Eq \define \{ I_{N \times N} : N \in \bN
\}$ where $I_{N \times N}$ denotes the $N \times N$ identity matrix. 
\end{definition}

In other words, for input size $N$, Alice and Bob receive elements $x,y \in [N]$ and wish to decide
whether $x = y$. It is well-known that $\R(\Eq) = 2$.

Constant-cost communication reductions, specifically to the \textsc{Equality} problem, have been
used implicitly in several prior works. Here we choose to explicitly define \emph{constant-cost
reductions} in general\footnote{This general definition of constant-cost reductions has arisen out
discussions with several other researchers.}. For this, we require the notion of a \emph{query set}.

\newcommand{\QS}{\mathsf{QS}}
\begin{definition}
A \emph{query set} $\cQ$ is a set of matrices that is closed under the following operations:
\begin{enumerate}
\item For every $Q \in \cQ$ and any $Q'$ obtained by row and column permutations of $Q$, $Q' \in \cQ$.
\item For every $Q \in \cQ$, if $Q'$ is any submatrix of $Q$ then $Q' \in \cQ$.
\item For every $Q \in \cQ$, if $Q'$ is obtained by duplicating a row or a column of $Q$, then $Q'
\in \cQ$.
\end{enumerate}
For a set $\cP$ of matrices, we define $\QS(\cP)$ to be the closure of $\cP$ under these operations.
\end{definition}
In the communication complexity literature, $\QS(\Eq)$ was recently named the set of \emph{blocky
matrices} \cite{HHH22dimfree,AY22}. In graph theory, $\QS(\Eq)$ are the adjacency matrices of
disjoint unions of bicliques, also called \emph{bipartite equivalence graphs}. It is easily verified
that for any constant $c$, if $\R(\cP) \leq c$ then $\R(\QS(\cP)) \leq c$. However, we caution that
$\R(\QS(\cP)) \leq \R(\cP)$ does not hold for \emph{non-constant} complexities, because $\QS(\cP)$
includes all submatrices of $\cP$ and $\R(\cdot)$ takes the \emph{maximum} complexity over all
size-$N$ matrices (see \cite{HWZ22} for examples).

We now give two equivalent definitions for reductions between problems; one algorithmic and one
structural.
\newcommand{\D}{\mathsf{D}}
\begin{definition}[Communication with Oracles]
Let $\cB$ be a communication problem and let $P \in \pmset^{N \times N}$. A \emph{deterministic
protocol computing $P$ with $\cB$ oracles} is a rooted binary tree $T$ where each leaf $\ell$ is
assigned a value $b(\ell) \in \pmset$ and inner node $v$ is assigned an $N \times N$ matrix $Q_v \in
\QS(\cB)$, with the following conditions. On each pair of inputs $x,y \in [N]$ the protocol begins
at the root node $v$ of $T$. At each node $v$, if $Q_v(x,y) = -1$ then the protocol proceeds by
advancing the current node $v$ to its left child, and if $Q_v(x,y) = 1$ then the protocol proceeds
by advancing the current node $v$ to its right child, until $v$ becomes a leaf, at which point the
protocol outputs $b(v)$. It is required that $b(v) = P_{x,y}$ for all inputs $x,y$.

The \emph{cost} of the protocol is the depth of the tree. We write $\D^{\cB}(P)$ for the minimum
cost of a protocol which computes $P$ with $\cB$ oracles. For a communication problem $\cA$, we write $\D^\cB(\cA)$ for
the function
\[
N \mapsto \max_{P \in \cA_N} \D^\cB(P).
\]
\end{definition}

In other words, a communication protocol with $\cB$ oracles is a \emph{deterministic} protocol where
in each round, Alice and Bob transform their inputs $x,y$ into inputs to a problem in $\cB$ and
receive the answer from an oracle computing $\cB$ at unit cost. Observe that, as long as $\cB$ is
non-trivial (\ie does not contain only all-$1$ and all-($-1$) matrices), the definition of
$\QS(\cB)$ allows any single round of deterministic communication to be simulated by an oracle, so
without loss of generality we may assume that every inner node of the protocol is an oracle call.

If there is a constant $c$ such that $\D^\cB(\cA) \leq c$, then we say that $\cA$
\emph{constant-cost reduces} (or just \emph{reduces}) to $\cB$. The following proposition is easily
obtained by standard error-boosting techniques:

\begin{proposition}
\label{prop:const-cost-reduction}
Suppose $\R(\cB) = O(1)$ and $\cA$ reduces to $\cB$. Then $\R(\cA) = O(1)$. In particular, if $\cA$
reduces to $\Eq$ then $\R(\cA) = O(1)$.
\end{proposition}

The second, structural definition of reduction is as follows. We say $\cA$ reduces to $\cB$ if there
exists a constant $t$ such that, for every $A \in \cA$, there exists:
\begin{enumerate}
\item a function $f : \pmset^t \to \pmset$; and
\item matrices $Q_1, \dotsc, Q_t \in \QS(\cB)$,
\end{enumerate}
such that $A = f(Q_1, \dotsc, Q_t)$, meaning that $A(i,j) = f(Q_1(i,j), Q_2(i,j), \dotsc, Q_t(i,j))$
for all $i,j \in [N]$. In the special case when $\cB$ is the set of identity matrices, this
definition appeared independently in \cite{HHH22dimfree,HWZ22} and subsequently in \cite{EHK22}, and
the minimum $t$ such that the above conditions hold is a ``functional'' analog of rank, recently
called the \emph{functional blocky-rank} in \cite{AY22}. It is not difficult to show that this
structural definition of constant-cost reductions is equivalent to the algorithmic one. One may
easily derive a constant-cost protocol with oracles $Q_i$ from the structural definition, and in the
other direction one may simply let the set of matrices $Q_i$ be the inner nodes of the communication
protocol and define $f$ as the function that simulates the protocol on these queries. In the
structural definition it is not hard to see an analog of \cref{prop:const-cost-reduction} for
implicit representations. A similar\footnote{There are some technicalities involved in translating
between the two.} notion of reductions for implicit representations appeared independently and
concurrently in \cite{Chan23}, which included reductions to \textsc{Equality} and
\textsc{Greater-Than} as parts of a complexity hierarchy of implicit representations.

\begin{proposition}
\label{prop:const-cost-reduction-implicit}
Suppose $\cB$ is the set of adjacency matrices for a hereditary graph class that admits an implicit
representation, and suppose $\cA$ is the set of adjacency matrices for a hereditary graph class
$\cG$. If $\cA$ reduces to $\cB$ then $\cG$ admits an implicit representation.
\end{proposition}

\subsection{From Communication Protocols to Implicit Representations}
\label{section:communication-to-labeling}

An observation of \cite{Har20,HWZ22} is that any hereditary graph class $\cG$ for which $\R(\cG) =
O(1)$ must also have an implicit representation (and any constant-cost communication problem may be
transformed into a hereditary graph class). Therefore, as argued in \cite{HWZ22}, constant-cost
communication is essentially the probabilistic version of implicit representations.

We will present our proofs as upper bounds on communication complexity, which imply implicit
representations.  The general correspondence between constant-cost communication and implicit
representations is non-constructive (by the probabilistic method), but for the sake of clarity and
completeness, we briefly describe how to directly translate a communication protocol that uses
\textsc{Equality} oracles (as ours will do) into an implicit representation.  

Recall that, for a graph $G = (V,E)$, if $\DEQ(G) \leq c$ then there exists a binary communication
tree of depth $c$ with each inner node $v$ assigned to a matrix $Q_v \in \QS(\Eq)$, which means that
$Q_v$ is the adjacency matrix of a bipartite equivalence graph. In other words, there are functions
$a_v, b_v : V \to [N]$ such that
\[
Q_v(x,y) = \begin{cases}
  1 &\text{ if } a_v(x) = b_v(y) \\
  0 &\text{ otherwise.}
  \end{cases}
\]
To obtain an implicit representation, we need to define a decoder $D$ and encodings $\enc(\cdot)$
for each graph $G \in \cG$.  We define $\enc(x)$ for each $x \in V$ by writing down the values
$a_v(x), b_v(y)$ for each inner node $v$ of the tree, together with the output values at the leaves
of the tree. Each value $a_v(x)$ and $b_v(x)$ requires at most $\lceil \log N \rceil$ bits, and
there are at most $2^c$ nodes in the tree, which is constant, so the size of the encoding is $O(\log
N)$. The decoder $D$, on inputs $\enc(x)$ and $\enc(y)$ for $x,y \in V$, may use the values of
$a_v(x)$ and $b_v(y)$ for each node $v$, together with the outputs on the leaves, to simulate the
communication protocol.

\section{Communication Bounds for Excluded Cycles and Subdivided Stars}
\label{section:gyarfas}

Our results for unit disk graphs and matrices of sign-rank 3 will follow from a more general result
on bipartite graphs excluding either long cycles or subdivided stars, which we prove in this
section. Recall the definition of the subdivided star $S_{s, t}$, \cref{def:subdivided-star}.

\newcommand{\adj}{\mathsf{adj}}

\thmintromain*

Our main tool will be the \emph{\Gyarfas decomposition}, which we borrow from \cite{POS22}, defined
below.

\subsection{\Gyarfas Decomposition: Definition, Existence, and Properties}
\label{sec:gyarfas-def}

\newcommand{\depth}{\mathsf{depth}}
\newcommand{\bag}{\mathsf{bag}}
The following definition of the \Gyarfas decomposition is taken from \cite{POS22}.  We will only
apply the \Gyarfas decomposition to bipartite graphs in this paper, so we state the special case of
the decomposition for bipartite graphs. See \cref{fig:Gyarfas} for an illustration.

\begin{definition}[\Gyarfas Decomposition]
\label{def:gyarfas}
A \emph{\Gyarfas decomposition} of a connected bipartite graph $G$ is a rooted tree $Y$ satisfying the
following properties:
\begin{enumerate}
\item Each node of $Y$ is a subset of $V(G)$, called a \emph{bag}, and the nodes of $Y$ form a
  partition of $V(G)$. For each vertex $v \in V(G)$, write $\bag_Y(v)$ for the unique bag in $Y$ that
  contains $v$. We will drop the subscript $Y$ when the decomposition is clear from context.
\item The root bag of $Y$ is a singleton containing the \emph{root vertex}.
\item If $u,v \in V(G)$ are adjacent then $\bag(u)$ is an ancestor of $\bag(v)$ or vice-versa.
\item For every bag $B$ of $Y$, the subgraph of $G$ induced by $B$ together with all of its
descendents is connected.
\item For every non-root bag $B$ of $Y$, there exists a vertex $h(B)$, called the \emph{hook} of
$B$, which belongs to the parent bag of $B$ and has the property that $h(B)$ is adjacent to
all vertices of $B$ and non-adjacent to all vertices in the strict descendents of $B$.
\end{enumerate}
For each bag $B$, we write $\depth(B)$ for the length of the path from the root bag to $B$ in $Y$
(where the depth of the root bag is 0).  For each $\ell \in \bN$, we say that \emph{level $\ell$} of
$Y$ is the set of all bags $B$ with $\depth(B) = \ell$.  A \Gyarfas decomposition for a disconnected
bipartite graph $G$ is the union of \Gyarfas decompositions for its connected components.
\end{definition}

\begin{figure}[th]
	\centering
	\scalebox{0.8}{%
	\begin{tikzpicture}
		\centering
		
		\def\rscoef{1.5}
		 
		 % Bag style
		\tikzset{bag/.style={draw, ellipse, fill=white, minimum width=2cm, minimum height=1cm}}
		\tikzset{bag-edge/.style={draw=gray!50}}
		\tikzset{hook-edge/.style={draw,BrickRed,line width=1.1pt}}

		% First column
		\foreach \i in {2,...,6} {
			\node[bag, label=left:$B_\i$] (O1\i) at (0, \rscoef*7-\rscoef*\i) {};
			\node[w_s_vertex, label={[xshift=10pt, yshift=-4pt]\small{$h_\i$}}] (h\i) at (-0.3, \rscoef*7-\rscoef*\i) {};
		}
	
		\node[bag, label=left:$B_1$] (O11) at (0, \rscoef*7-\rscoef*1) {};
		\node[w_s_vertex, label={[xshift=-5pt, yshift=-4pt]\small{$h_1$}}] (h1) at (-0.3, \rscoef*7-\rscoef) {};
		
		\foreach \i in {1,...,5} {
			\pgfmathtruncatemacro{\j}{\i+1}
			
			% hook-bag edges
			\draw[bag-edge] (h\i) -- ($(O1\j.west) + (+0.04,0.13)$);
			\draw[bag-edge] (h\i) -- ($(O1\j.east) + (-0.04,0.15)$);
			
			% hook-hook  edges
			\draw[hook-edge] (h\i) -- (h\j);
		}
		\node[w_vertex, label=above:$r$] (V) at (4, \rscoef*7.5) {};
		\draw[hook-edge] (V) -- (h1);
		
		% Bag-bag edges
		\draw (O15.east) .. controls +(0.8,3) .. (O12.east);
		\draw (O14.east) .. controls +(0.4,1.5) .. (O11.east);
		\draw (O16.east) .. controls +(0.8,3) .. (O13.east);
		\draw (O16.east) .. controls +(1.2,3.2) .. (O11.east);

		% Second column
		\node[bag] (O22) at (4, \rscoef*6) {};
		\node at (4, \rscoef*5) {$\vdots$};
		
		% Third column
		\node[bag] (O33) at (8, \rscoef*6) {};
		\node at (8, \rscoef*5) {$\vdots$};
		
	   % Vertex and edges
		\draw[bag-edge] (V) -- ($(O11.west) + (0.5,0.43)$);
		\draw[bag-edge] (V) -- (O11.east);
		
		\draw[bag-edge] (V) -- ($(O22.west) + (0.01,0.08)$);
		\draw[bag-edge] (V) -- ($(O22.east) + (-0.01,0.08)$);
		
		\draw[bag-edge] (V) -- (O33.west);
		\draw[bag-edge] (V) --  ($(O33.east) + (-0.5,0.43)$);

		% Intersections
		%\path[name path=ellipse1] (0,6) ellipse (1cm and 0.5cm);
		
		% Connections to ellipse boundaries
		%\draw[-] (V) -- (intersection of V--O11.west and ellipse1);
	\end{tikzpicture}
	}
	\caption{\Gyarfas decomposition. The vertices $r$ and $h_i$, $i \in [6]$, are hooks. The hook
vertices induce a hook path highlighted with bold edges. Lines between two bags indicate that there
are edges between some vertices in those bags.}
	\label{fig:Gyarfas}
\end{figure}

There is a simple algorithmic proof that such decompositions always exist \cite{POS22}.

\begin{lemma}
For every connected bipartite graph $G$ and vertex $r \in V(G)$, there exists a \Gyarfas
decomposition of $G$ with root vertex $r$. Given $G$ and $r$, this decomposition can be computed in
polynomial time.
\end{lemma}

A path $P = (v_0, v_1, v_2, \ldots, v_k)$ in $G$ is a \emph{hook path} (with respect to $Y$) if $v_i$ is the hook
of $\bag(v_{i-1})$ for every $i \in [k]$.
Observe that any hook path with respect to a \Gyarfas decomposition is an induced path.
\Gyarfas decompositions are typically used in the case where some induced path $P_t$ is forbidden,
in which case the depth of the decomposition is bounded. In our case, we will not necessarily have a forbidden
$P_t$ or bounded depth of the decomposition, but we will see that the decomposition has a different
structure that will permit efficient communication protocols. For this we define the notion of
\emph{back degree}.

\begin{definition}[Back-Degree]
Given a \Gyarfas decomposition $Y$ of $G$. We say that a bag $B$ of $Y$ has an edge to another bag
$B'$ in $Y$ if there exist a vertex in $B$ and a vertex in $B'$ that are adjacent.  The
\emph{back-degree} of a bag $B$ in $Y$ is the number of ancestor bags of $B$ to which $B$ has an
edge.  The \emph{maximum back-degree} of $Y$ is the maximum back-degree of any of its bags.
\end{definition}

Note that \Gyarfas decomposition of a $P_t$-free graph has depth at most $t$, and therefore the
maximum back-degree of the decomposition is also bounded by $t$.

In the next two sections we show that if a graph has bounded chain-index and either
\begin{enumerate}
\item does not contain long induced cycles (\cref{sec-chordality}), or
\item does not contain a fixed subdivision of a star (\cref{sec-sub-star}),
\end{enumerate}
then its \Gyarfas decompositions have bounded maximum back-degree. In
\cref{section:protocol-for-gyarfas}, we give a general communication protocol for \Gyarfas
decompositions with bounded maximum back-degree.

Before proceeding \cref{sec-chordality} and \cref{sec-sub-star}, we introduce some notation and
properties of the interactions between bags in \Gyarfas decompositions that are used in both
sections.  Let $Y$ be a \Gyarfas decomposition of a bipartite graph $G = (X,Y,E)$, and let $B$ be a
bag of $Y$ with $\depth(B) > 0$. Write $h$ for the hook of $B$. Let $A_1, A_2, \dotsc, A_r$ be some
ancestors of $B$, excluding the immediate parent of $B$, to which $B$ has an edge. Then the
following properties are easy to verify:

\begin{proposition}
Let $s \in \bN$ and suppose that $\depth(A_1) < \depth(A_2) < \dotsm < \depth(A_r)$ and
$\depth(A_{i+1}) - \depth(A_{i}) \geq s$ for all $i \in [r-1]$.  For $i \in [r]$, we define
$h_{i,1}$ to be the hook of $A_i$, and for $z \in [s-1]$, inductively define $h_{i,z}$ as the hook
of $\bag(h_{i,z-1})$. For each $i \in [r]$, let $a_i \in A_i$ be a neighbour of some $b_i \in B$.
Then the following properties hold:
\begin{enumerate}[label=(\arabic*)]
	\item The hook $h$ of $B$ is adjacent to each $b_i$. For each $i \in [r]$ and $z \geq 1$, $a_i$ is
	not adjacent to $h$, because they are on the same side of the bipartition of $G$, and $h_{i,z}$ is
	not adjacent to $h$, because $h_{i,z}$ is a hook in an ancestor bag of $\bag(h)$ that is not the
	parent of $\bag(h)$.
	\label{item:prop-1}
	
	\item For each $i,j \in [r]$, $a_i$ is not adjacent to $a_j$, because they are on the same side of the bipartition of $G$.
	\label{item:prop-aiaj}
	
	\item For each $1 \leq i < j \leq r$ and each $z \geq 1$, $h_{i,z}$ is not adjacent to $a_j$, because $h_{i,z}$ is a hook that is not in the parent bag of $A_j$.
	%: the bags $A_i$ and $A_j$ are separated by at least $t-3$ levels of $Y$.
	\label{item:prop-ah}
	
	\item For each $i \in [r]$ and each $z \geq 2$, we have $h_{i,z}$ not adjacent to $a_i$ because $h_{i,z}$ is a
	hook that is not in the parent bag of $\bag(a_i)$.
	\label{item:prop-ahi}
	
	\item For each $i,j \in [r]$, and $z \geq 1$, $h_{i,z}$ is not adjacent to $b_j$ because $h_{i,z}$ is a hook
	that is not in the parent bag of $B$.
	\label{item:prop-bh}
\end{enumerate}
\end{proposition}

\subsection{Excluding Long Cycles}\label{sec-chordality}

\begin{lemma}
	\label{lemma:gyarfas-bdd-chordality}
	For any $t, k \in \bN$, there exists a constant $\ell$ such that the following holds. 
	Let $G = (X,Y,E)$ be any $(C_t, C_{t+1}, C_{t+2}, \ldots)$-free bipartite graph with $\ch(G) < k$.
Let $Y$ be a \Gyarfas decomposition of $G$. Then $Y$ has maximum back-degree at most $\ell$.
\end{lemma}
\begin{proof}
	Without loss of generality we assume that $t \geq 4$.
	Let $R$ be the Ramsey number that guarantees that a complete graph on $R$ vertices with edges
	colored by $2^{t-4}$ colors has a monochromatic clique of size $r \define \max\{ 2, k \}$.
	
	Let $\ell = (t-3) \cdot R$ and let $B$ be a bag of $Y$. If $B$ has depth at most $\ell$ in $Y$ the result
	holds trivially, so we will assume that $B$ has depth greater than $\ell$. Let $A''_1, A''_2, \dotsc,
	A''_m$ be the ancestors of $B$, excluding the immediate parent of $B$, to which $B$ has an edge.
	For each $i \in [m]$, let $a''_i \in A''_i$ be a neighbour of some $b''_i \in B$.
	
	Assume for the sake of contradiction that $B$ has edges to more than $\ell$ ancestors, so $m \geq \ell$.
	Then there is a subsequence of ancestor bags $A'_1, \dotsc, A'_R$ such that
	$\depth(A'_{i+1}) - \depth(A'_i) \geq t-3$, for each $i \in [R-1]$, so that there are at least $t-4$ levels of the \Gyarfas decomposition separating each bag in this subsequence. We will write $a'_i \define a''_{i^*}$ and
	$b'_i \define b''_{i^*}$, where $i^*$ is the index of the bag satisfying $A''_{i^*} = A'_i$.
	
	For each $A'_i$, let $h'_{i,1}$ be the hook of $A'_i$, and for $1 < z \leq t-3$, inductively define
	$h'_{i,z}$ as the hook of $\bag(h'_{i,z-1})$.
	For each pair $\{ A'_i, A'_j \}$ with $i < j$ we assign a color $\mathsf{col}\{A'_i, A'_j\}
	\in \zo^{t-4}$ as follows: 
	\begin{enumerate}
		\item[] the $z^{th}$ bit is 1 if and only if $a_i'$ is adjacent to $h_{j,z}'$, for $z \in [t-4]$.
	\end{enumerate}

	By Ramsey's theorem, we may now choose a subsequence $A_1, \dotsc, A_r$ of ancestor bags, where for
	each $i \in [r]$ there is a corresponding $i^* \in [R]$ such that $A_i = A'_{i^*}$, and each
	pair $\{A_i, A_j\}$ with $1 \leq i < j \leq r$ has the same color. We will now obtain a contradiction for each
	possibility of this color. We will write $a_i \define a'_{i^*}$, $b_i \define b'_{i^*}$, and
	$h_{i,z} \define h'_{i^*, z}$, for each $i \in [r]$ and $z \in [t-3]$, and use the notation and the properties from \cref{sec:gyarfas-def}.
	
%	Before deriving the contradictions, we observe the following properties of the subsequence:
%	\begin{enumerate}
%		\item Let $h$ be the hook of $B$. Then $h$ is adjacent to each $b_i$. For each $i$ and $z$, $a_i$ is
%		not adjacent to $h$, because they are on the same side of the bipartition of $G$, and $h_{i,z}$ is
%		not adjacent to $h$, because $h_{i,z}$ is a hook in an ancestor bag of $\bag(h)$ that is not the
%		parent of $\bag(h)$.
%		\label{item:prop-1}
%		\item For each $i$, $j$, $a_i$ is not adjacent to $a_j$, because they are on the same side of the bipartition of $G$.
%		\label{item:prop-aiaj}
%		\item For each $i < j$ and each $z$, $h_{i,z}$ is not adjacent to $a_j$, because $h_{i,z}$ is a hook
%		that is not in the parent bag of $A_j$.
%		%: the bags $A_i$ and $A_j$ are separated by at least $t-3$ levels of $Y$.
%		\label{item:prop-ah-cycle}
%		\item For each $i$ and each $z \geq 2$, we have $h_{i,z}$ not adjacent to $a_i$ because $h_{i,z}$ is a
%		hook that is not in the parent bag of $\bag(a_i)$.
%		\label{item:prop-ahi}
%		\item For each $i$, $j$, and $z$, $h_{i,z}$ is not adjacent to $b_j$ because $h_{i,z}$ is a hook
%		that is not in the parent bag of $B$.
%		\label{item:prop-bh}
%	\end{enumerate}
%	We now consider the color of the monochromatic sequence.
	
	\noindent
	\textbf{Case 1:} There is $z \in [t-4]$ such that the $z^{th}$ bit of the color is 1. 
	Consider the subgraph $H$ induced by the vertices
	$\{a_1, a_2, \dotsc, a_r\} \cup \{ h_{1,z}, h_{2,z}, \dotsc, h_{r,z} \}$. For $1 \leq i < j \leq r$, we have $a_i$
	adjacent to $h_{j,z}$ due to the color, and $h_{i,z}$ is not adjacent to $a_j$ due to Property \ref{item:prop-ah}. Thus, by definition, $\ch(G) \geq \ch(H) = r \geq k$, a contradiction.
	
	\noindent
	\textbf{Case 2:} All bits of the color are 0. Consider the hook path $P$ from $a_2$ to $h_{1,1}$.
	Let $v$ be the first (i.e. closest to $a_2$) vertex on $P$ that is adjacent to $a_1$. Such a vertex exists because $a_1$ is adjacent to the last vertex $h_{1,1}$ of the path.
	Let $P'$ be the subpath of $P$ from $a_2$ to $v$.
	By Property \ref{item:prop-aiaj} and the color assumption, $P'$ contains the first $t-2$ vertices of $P$:
	$a_2, h_{2,1}, h_{2,2}, \ldots, h_{2,t-4}, h_{2,t-3}$.
	Now, if $b_2$ is adjacent to $a_1$, then $b_2,P',a_1,b_2$ is an induced cycle of length at least $t$.
	Similarly, if $b_1$ is adjacent to $a_2$, then $b_1, P', a_1, b_1$ is such a cycle. 
	Finally, if neither $b_2$ is adjacent to $a_1$, nor $b_1$ is adjacent to $a_2$, then $b_1 \neq b_2$ and,
	by Properties \ref{item:prop-1} and \ref{item:prop-bh}, $h,b_2,P',a_1,b_1,h$ is a forbidden induced cycle.
\end{proof}

\subsection{Excluding Subdivisions of Stars}
\label{sec-sub-star}

%Recall the definition of the subdivided star $S_{s, t}$, \cref{def:subdivided-star}.

\begin{lemma}
\label{lemma:gyarfas-star}
For any $s, t, k \in \bN$, there exists a constant $\ell$ such that the following holds. Let $G =
(X,Y,E)$ be any bipartite graph with $\ch(G) < k$ that does not contain $S_{s, t}$ as an
induced subgraph. Let $Y$ be a \Gyarfas decomposition of $G$. Then $Y$ has maximum back-degree at
most $\ell$. 
\end{lemma}
\begin{proof}
Let $R$ be the Ramsey number that guarantees that a complete graph on $R$ vertices with edges
colored by $2^{3+(t-1)}$ colors has a monochromatic clique of size $r \define \max\{ s, k \}$.

Let $\ell = t \cdot R + 1$ and let $B$ be a bag of $Y$. If $B$ has depth at most $\ell$ in $Y$ the result
holds trivially, so we will assume that $B$ has depth greater than $\ell$. Let $A''_1, A''_2, \dotsc,
A''_m$ be the ancestors of $B$, excluding the immediate parent of $B$, to which $B$ has an edge,
meaning that for each $i \in [m]$, there exists a vertex $b''_i \in B$ with an edge to a vertex $a''_i
\in A''_i$.

Assume for the sake of contradiction that $B$ has edges to more than $\ell$ ancestors, so $m \geq \ell$.
Then there is a subsequence of ancestor bags $A'_1, \dotsc, A'_R$ such that $\depth(A'_1) \geq t$ and $\depth(A'_{i+1}) - \depth(A'_i) \geq t$, for each $i \in [R-1]$; in particular there are at least $t-1$ levels of the Gy\'arf\'as decomposition separating each bag in this subsequence. We will write $a'_i \define a''_{i^*}$ and
$b'_i \define b''_{i^*}$, where $i^*$ is the index of the bag satisfying $A''_{i^*} = A'_i$.

For each $A'_i$, let $h'_{i,1}$ be the hook of $A'_i$, and for $1 < z \leq t-1$, inductively define
$h'_{i,z}$ as the hook of $\bag(h'_{i,z-1})$.
For each pair $\{ A'_i, A'_j \}$ with $i < j$ we assign a color $\mathsf{col}\{A'_i, A'_j\}
\in \zo^{3+(t-1)}$ as follows:
\begin{enumerate}
\item The first bit indicates whether $b'_i = b'_j$ (i.e. set the bit to 1 if $b'_i = b'_j$ and 0
otherwise).
\item The second bit indicates whether $b'_i$ is adjacent to $a'_j$.
\item The third bit indicates whether $b'_j$ is adjacent to $a'_i$.
\item The remaining $t$ bits indicates whether $a'_i$ is adjacent to $h'_{j,z}$, for $z \in [t-1]$.
\end{enumerate}
By Ramsey's theorem, we may now choose a subsequence $A_1, \dotsc, A_r$ of ancestor bags, where for
each $i \in [r]$ there is a corresponding $i^* \in [R]$ such that $A_i = A'_{i^*}$, and each
pair $\{A_i, A_j\}$ with $i < j$ has the same color. We will now obtain a contradiction for each
possibility of this color. We will write $a_i \define a'_{i^*}$, $b_i \define b'_{i^*}$, and
$h_{i,z} \define h'_{i^*, z}$, for each $i \in [r]$ and $z \in [t-1]$, and use the notation and the properties from \cref{sec:gyarfas-def}.

%Before deriving the contradictions, we observe the following properties of the subsequence:
%\begin{enumerate}
%\item Let $h$ be the hook of $B$. Then $h$ is adjacent to each $b_i$. For each $i$ and $z$, $a_i$ is
%not adjacent to $h$, because they are on the same side of the bipartition of $G$, and $h_{i,z}$ is
%not adjacent to $h$, because $h_{i,z}$ is a hook in an ancestor bag of $\bag(h)$ that is not the
%parent of $\bag(h)$.
%\item For each $i < j$ and each $z$, $h_{i,z}$ is not adjacent to $a_j$, because $h_{i,z}$ is a hook
%that is not in the parent bag of $A_i$: the bags $A_i$ and $A_j$ are separated by at least $2t$
%levels of $Y$.
%\label{item:prop-ah}
%\item For each $i$ and each $z \geq 2$, we have $h_{i,z}$ not adjacent to $a_i$ because $h_{i,z}$ is a
%hook that is not in the parent bag of $\bag(a_i)$.
%\label{item:prop-ahi}
%\item For each $i$, $j$, and $z$, $h_{i,z}$ is not adjacent to $b_j$ because $h_{i,z}$ is a hook
%that is not in the parent bag of $B$.
%\label{item:prop-bh}
%\end{enumerate}
%We now consider the color of the monochromatic sequence.

\noindent
\textbf{Case 1:} There is $z \in [t-1]$ such that the $(3+z)^{th}$ bit of the color is 1. The argument is exactly as in Case 1 of \cref{lemma:gyarfas-bdd-chordality}.
Consider the subgraph $H$ induced by the vertices
$\{a_1, a_2, \dotsc, a_r\} \cup \{ h_{1,z}, h_{2,z}, \dotsc, h_{r,z} \}$. For $1 \leq i < j \leq r$, we have $a_i$
adjacent to $h_{j,z}$ due to the color, and $h_{i,z}$ is not adjacent to $a_j$ due to Property \ref{item:prop-ah}. Thus, by definition, $\ch(G) \geq \ch(H) = r \geq k$, a contradiction.

\noindent
\textbf{Case 2:} The first bit or second bit of the color is 1, and the $(3+z)^{th}$ color is 0 for
all $z \in [t-1]$. 
Consider the subgraph induced by the vertices
\[
\{b_1\} \cup \bigcup_{i=1}^{s} \{a_i, h_{i,1}, h_{i,2}, \dotsc, h_{i,t-1} \} \,.
\]
Since each $A_i$ is separated by at least $t-1$ levels of $Y$, each of the above named vertices are
distinct.  If the first bit of the color is 1, then we have $b_1$ adjacent to each $a_i$ by
definition, since $b_1 = b_2 = \dotsm = b_s$. If the first bit of the color is 0 but the second bit
of the color is 1, then we have $b_1$ adjacent to each $a_i$ because of the color.

For each $1 \leq i < j \leq s$ and $z \in [t-1]$, we have $a_i$ not adjacent to $h_{j,z}$ because
the associated bit of the color is set to 0, and we have $a_j$ not adjacent to $h_{i,z}$ by Property
\ref{item:prop-ah}. For $z \geq 2$, we have $a_i$ not adjacent to $h_{i,z}$ by Property
\ref{item:prop-ahi}. We have $a_i$ not adjacent to $a_j$ by Property
\ref{item:prop-aiaj}. And we have $b_1$ not adjacent to $h_{i,z}$ by Property \ref{item:prop-bh}. But we have
$b_1$ adjacent to each $a_i$, as well as edges $a_ih_{i,1}$ and $h_{i,1}h_{i,2}$, $h_{i,2}h_{i,3}$,
$\dotsc$, $h_{i,z-1}h_z$ by definition. So the subgraph induced by the considered vertices is $S_{s, t}$, which is a contradiction.

\noindent
\textbf{Case 3:} The first two bits of the color are 0, the third bit is 1, and the $(3+z)^{th}$ bit
is 0 for each $z \in [t-1]$. Consider the subgraph $H$ induced by the vertices $\{b_1, \dotsc, b_k\}
\cup \{ a_1, \dotsc, a_k \}$. For each $i < j$, $a_i$ is adjacent to $b_j$ due to the third bit of
the color, but $b_i$ not adjacent to $a_j$ due to the second bit of the color. Then we have $\ch(G) \geq \ch(H) = r \geq k$, a contradiction.

\noindent
\textbf{Case 4:} All bits of the color are 0. Consider the subgraph induced by the vertices
\[
  \{ h \} \cup \bigcup_{i=1}^s \{ b_i, a_i, h_{i,1}, \dotsc, h_{i,t-2} \} \,.
\]
Since each bag $A_i$ in the sequence is separated by at least $t-1$ levels of $Y$ and the first bit of the color is 0, each of the named
vertices above is distinct. 
By Property \ref{item:prop-1}, $h$ is adjacent to none of the vertices
$a_i, h_{i,1}, \dotsc, h_{i,t-2}$ for every $i \in [s]$.
For each $i < j$, we have $b_i$ not adjacent to $a_j$ and $a_i$ not
adjacent to $b_j$ due to the color. For each $z \in [t-1]$, we have $h_{i,z}$ not adjacent to $b_i$
or $b_j$ due to Property \ref{item:prop-bh}; and we have $h_{i,z}$ not adjacent to $a_j$ due to Property
\ref{item:prop-ah} and $a_i$ not adjacent to $h_{j,z}$ due to the color. For $z \geq 2$ we have
$h_{i,z}$ not adjacent to $a_i$ due to Property \ref{item:prop-ahi}.

On the other hand, we have edges $hb_i$ for each $i \in [s]$ by definition, along with edges
$b_ia_i$, $a_ih_{i,1}$, and $h_{i,1}h_{i,2}$, $\dotsc$, $h_{i,t-3}h_{i,t-2}$. Therefore the induced subgraph is $S_{s, t}$, which is a contradiction.
\end{proof}

\subsection{A Communication Protocol for the Gy\'arf\'as Decomposition}
\label{section:protocol-for-gyarfas}

Let $G$ be a connected bipartite graph and let $Y$ be a \Gyarfas decomposition of $G$.
For any bag $B$ of $Y$ with depth $d = \depth(B)$, let $G_B$ denote the subgraph of $G$ induced by
$B$ together with all of the descendent bags of $B$ in $Y$ with depth $d' \not\equiv d \mod 2$.
We require the next two lemmas of \cite{POS22}.

\begin{lemma}
\label{lemma:pos-depth-2}
Let $B$ be a bag of $Y$ with $\depth(B) = d$ for any $d \geq 2$. Then $\ch(G_B) < \ch(G)$.
\end{lemma}

\begin{lemma}
\label{lemma:pos-depth-1}
Let $B$ be a bag of $Y$ with $\depth(B) = 1$. Let $C$ be a connected component of $\bc{G_B}$ and
$Y_C$ be a \Gyarfas decomposition of $C$ rooted at a vertex $r_C \in V(C) \cap B$. Let $B'$ be a bag
of $Y_C$ with $\depth(B') = d \geq 1$. Then $\ch(C_{B'}) < \ch(G)$.
\end{lemma}

We will also require the following easy fact.
\begin{proposition}
\label{proposition:base-case}
Let $\cG$ be the class of bipartite graphs $G$ with $\ch(G) = 1$. Then $\DEQ(\cG) \leq 2$.
\end{proposition}
\begin{proof}
Note that the chain-index of $P_4$, the 4-vertex path, is 2.
Thus each $G \in \cG$ is $P_4$-free and therefore is a disjoint union of bicliques, \ie an equivalence graph.
Therefore Alice and Bob may compute adjacency in $G$ by using 1 bit of communication to ensure that
their input vertices $x$ and $y$ are on opposite sides of the bipartition, and using 1 call to the
$\Eq$ oracle to check if $x,y$ are in the same biclique.
\end{proof}

Our first main result, \cref{thm:intro-main}, follows from the next lemma, applied together
with \cref{lemma:gyarfas-bdd-chordality,lemma:gyarfas-star}.

\begin{lemma}
\label{lemma:gyarfas-protocol}
Let $\cG$ be a hereditary class of bipartite graphs that is closed under bipartite complementation, and which satisfies the following conditions:
\begin{enumerate}
\item There exists a constant $k$ such that $\ch(\cG) \leq k$.
\item There exists a constant $\ell$ such that for any $G = (X,Y,E) \in \cG$, any \Gyarfas decomposition
of $G$ has back-degree bounded by $\ell$.
\end{enumerate}
Then there exists a constant $c$ such that $\DEQ(\cG) \leq c$.
\end{lemma}
\begin{proof}
We prove the theorem by induction
on $k$. The base case $k = 1$ is established in \cref{proposition:base-case}. Let $x,y \in V(G)$ be
Alice's and Bob's inputs, respectively. We may assume without loss of generality that $G$ is
connected and that $x$ and $y$ are in opposite parts of the bipartition of $G$, since Alice and Bob
may use one $\Eq$ oracle call to check whether their inputs $x$ and $y$ are in the same connected
component, and use 1 bit of communication to determine whether $x$ and $y$ are in opposite parts.

Let $Y$ be a \Gyarfas decomposition of $G$. The communication protocol proceeds as follows. We will
assume that the root vertex of $Y$ is on the left side of the bipartition of $G$, and that Alice's
input $x$ is on the left side and Bob's input $y$ is on the right side of the bipartition.
\begin{enumerate}
\item Using 1 bit of communication, Alice tells Bob whether $x$ is the root vertex of $Y$. If so,
Bob outputs 1 if $y$ has depth 1 in $Y$ and the protocol terminates. The protocol is correct in this
case, since by \cref{def:gyarfas}, all vertices at depth 1 are adjacent to the root vertex.
\item Using 1 bit of communication, Bob tells Alice whether $y$ has depth 1 in $Y$. If so, they
perform the following:
  \begin{enumerate}
    \item Using 1 call to the $\Eq$ oracle, Alice and Bob decide if $\bag(x)$ is a descendent of
      $\bag(y)$. This is possible because Alice and Bob each know the set of level 1 bags of $Y$.
      If $\bag(x)$ is not a descendent of $\bag(y)$, they output 0 and the protocol terminates.
      The protocol is correct in this case, since by \cref{def:gyarfas}, if $x$ and $y$ are adjacent
      then $\bag(x)$ must be the descendent of $\bag(y)$ or vice versa.
    \item Alice and Bob now agree on $B = \bag(y)$, so they each compute the connected components
      $C_1, \dotsc, C_m$ of $\bc{G_B}$ and agree on \Gyarfas decompositions $Y_1, \dotsc,
      Y_m$ of these components, respectively, where the root vertex of each decomposition is on the
      right side of the bipartition. Using 1 call to the $\Eq$ oracle, they decide if $x,y$ belong to the
      same connected component of $\bc{G_B}$. If not, they output 1 and terminate the protocol. The
      protocol is correct in this case by definition.
    \item Let $i$ be the index of the component $C_i$ of $\bc{G_B}$ containing both $x$ and $y$.
      Using 1 bit of communication, Bob tells Alice whether $y$ is the root vertex of $Y_i$. If
      so, Alice outputs 0 if $x$ has depth 1 in $Y_i$ and the protocol terminates. The protocol is
      correct in this case, since by \cref{def:gyarfas} all vertices of depth 1 in $Y_i$ are adjacent
      to the root vertex $y$ in $\bc{G_B}$ (and therefore non-adjacent in $G$).
    \item By the assumption of bounded back-degree, $\bag_{Y_i}(x)$ has edges
      to at most $\ell$ of its ancestors in $Y_i$. Call these ancestors $A_1, \dotsc, A_{\ell'}$ where
      $\ell' \leq \ell$. Using $\ell$ calls to the $\Eq$ oracle, Alice and Bob determine whether $A_j =
      B'$ for some $j \leq \ell'$, where $B' \define \bag_{Y_i}(y)$.
      \label{item:protocol-step-comp1}
      \begin{enumerate}
        \item If $A_j = B'$ then Alice and Bob inductively compute adjacency in the graph
        $(C_i)_{B'}$, which is the bipartite complement of a graph in $\cG$ and therefore is
        contained in $\cG$, and which by \cref{lemma:pos-depth-1} satisfies $\ch((C_i)_{B'}) <
        \ch(G)$. They then output the opposite value and terminate the
        protocol. The protocol is correct in this case since, by induction, they will compute
        adjacency of $x,y$ in $(C_i)_{B'}$, which is an induced subgraph of $\bc{G}$, so $x,y$ have the
        opposite adjacency as in $G$.
        \item If $A_j \neq B'$ for all $j$, the protocol proceeds as below.
          \label{item:protocol-step-anc1}
      \end{enumerate}
    \item Similar to step \ref{item:protocol-step-comp1}, $\bag_{Y_i}(y)$ has edges to at most $\ell$
      of its ancestors in $Y_i$. Call these ancestors $A_1, \dotsc, A_{\ell'}$ where $\ell' \leq \ell$. 
      Using $\ell$ calls to the $\Eq$ oracle, Alice and Bob determine whether $A_j =
      B'$ for some $j \leq \ell'$, where $B' \define \bag_{Y_i}(x)$.
      \label{item:protocol-step-comp2}
      \begin{enumerate}
        \item If $A_j = B'$ then Alice and Bob inductively compute adjacency in the graph
        $(C_i)_{B'}$, which again is contained in $\cG$ and by \cref{lemma:pos-depth-1} satisfies
        $\ch((C_i)_{B'}) < \ch(G)$. They then output the opposite value and
        terminate the protocol. The protocol is correct in this case since, by induction, they will
        compute adjacency of $x,y$ in $(C_i)_{B'}$, which is an induced subgraph of $\bc{G}$, so
        $x,y$ have the opposite adjacency as in $G$.
        \item If $A_j \neq B'$ for all $j$, then Alice and Bob output 1 and the protocol terminates.
          The protocol is correct in this case because $x,y$ are adjacent in $G$ if and only if they are
          non-adjacent in $C_i$. By \cref{def:gyarfas}, if they are adjacent in $C_i$ then either
          $\bag_{Y_i}(x)$ is an ancestor of $\bag_{Y_i}(y)$ or vice versa. From step
          \ref{item:protocol-step-comp1}, we know that if $\bag_{Y_i}(y)$ is an ancestor of
          $\bag_{Y_i}(x)$, then $\bag_{Y_i}(x)$ has no edges to $\bag_{Y_i}(y)$, so $x,y$ are
          non-adjacent in $C_i$ and therefore adjacent in $G$. From the current step, we know
          similarly that if $\bag_{Y_i}(x)$ is an ancestor of $\bag_{Y_i}(y)$ then $x,y$ are again
          non-adjacent in $C_i$ and therefore adjacent in $G$.
          \label{item:protocol-step-anc2}
      \end{enumerate}
  \end{enumerate}
  \item Now guaranteed that $x$ and $y$ are each in bags at depth 2 or higher, Alice and Bob proceed
    similarly as in steps \ref{item:protocol-step-comp1} and \ref{item:protocol-step-comp2}, with the
    following differences. Here, $Y$ is used instead of $Y_i$. In step 
    \ref{item:protocol-step-anc2}, the protocol outputs 0 instead of 1, because they are operating
    on the graph $G$ itself instead of an induced subgraph of the bipartite complement $\bc{G}$.  When applying the inductive
    hypothesis, we use \cref{lemma:pos-depth-2} instead of \cref{lemma:pos-depth-1}, and the players do
    not flip the output of the protocol applied to $G_{B'}$. Correctness again follows by induction.
\end{enumerate}
This concludes the proof.
\end{proof}

\section{Application to Sign-Rank 3 and Unit Disk Graphs}

We now prove our results \cref{thm:intro-signrank,thm:intro-udg} for graphs of sign-rank 3 and
unit disk graphs. This will require the notion of \emph{edge-asteroid triples} (see e.g.
\cite{FHH99,STU10,AZ18}).

\begin{definition}[Edge-Asteroid Triples]
	A set of three edges in a graph is called an \emph{edge-asteroid triple} if for each pair of the edges, there is a
	path containing both of the edges that avoids the neighbourhoods of the end-vertices of the third edge (see \cref{fig:eat} for an illustration).
  We say that a graph class $\cG$ is \emph{edge-asteroid-triple-free} if no $G \in \cG$ contains an
  edge-asteroid triple.
\end{definition}

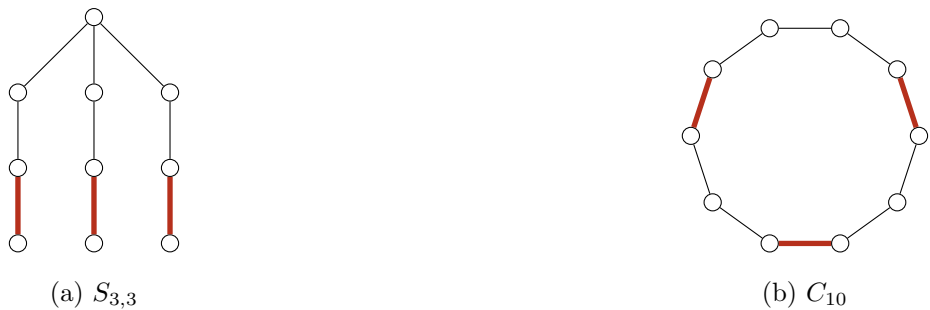
\begin{figure}[H]
	\begin{subfigure}[b]{0.45\textwidth}
		\centering
		% S_{3,3}
		\begin{tikzpicture}
			\node[w_vertex] (1) at (1,0) { }; 	
			\node[w_vertex] (2) at (0,-1) { };
			\node[w_vertex] (3) at (1,-1) { };
			\node[w_vertex] (4) at (2,-1) { };
			\node[w_vertex] (5) at (0,-2) { };
			\node[w_vertex] (6) at (1,-2) { };
			\node[w_vertex] (7) at (2,-2) { };
			\node[w_vertex] (8) at (0,-3) { };
			\node[w_vertex] (9) at (1,-3) { };
			\node[w_vertex] (10) at (2,-3) { };
			
			\foreach \from/\to in {1/2,1/3,1/4, 2/5,3/6,4/7}
			\draw (\from) -- (\to);
			
			\foreach \from/\to in {5/8,6/9,7/10}
			\draw[eat edge] (\from) -- (\to);
		\end{tikzpicture}
		\caption{$S_{3,3}$}
		\label{fig:subfig1}
	\end{subfigure}
	\hfill
	\begin{subfigure}[b]{0.45\textwidth}
		\centering
		% C_{10}
		\begin{tikzpicture}
			% Coordinates
			\foreach \i in {1,...,10} {
				\coordinate (P\i) at ({\i * 36}:1.5);
			}
			
			% Nodes
			\foreach \i in {1,...,10} {
				\node[w_vertex] (N\i) at (P\i) {};
			}
		
			% Edges
			\foreach \i in {1,...,9} {
				\draw(N\i) -- (N\the\numexpr\i+1\relax);
			}
			\draw (N10) -- (N1);
			
			\foreach \from/\to in {10/1, 4/5, 7/8} {
				\draw[eat edge] (N\from) -- (N\to);
			}
		\end{tikzpicture}
		\caption{$C_{10}$}
		\label{fig:subfig2}
	\end{subfigure}

	\caption{Both $S_{3,3}$ and $C_{10}$ contain an edge-asteroid triple (bold edges).}
	\label{fig:eat}
\end{figure}

Since $S_{3, 3}$ and $C_t$ for $t \geq 10$ contain edge-asteroid triples, we make the following
simple observation:

\begin{observation}
\label{obs:edge-asteroid-triple}
Let $G$ be any bipartite graph that is edge-asteroid-triple-free. Then $G$ is both $S_{3, 3}$-free,
and $C_t$-free for all $t \geq 10$.
\end{observation}

This observation will allow us to apply our \cref{thm:intro-main}, but it requires that our graphs
$G$ and their complements both be edge-asteroid-triple-free.  Unit disk graphs and graphs of
sign-rank 3 are not necessarily edge-asteroid-triple-free.  But we show that we can decompose these
graphs into pieces which satisfy the necessary conditions.

\subsection{Sign-Rank 3}

To apply \cref{thm:intro-main} to graphs of sign-rank 3, we will decompose these graphs into
pieces which are edge-asteroid-triple-free. We achieve this by interpreting graphs of sign-rank 3 as
point-halfspace incidences and projecting down into dimension 2. 
%From here, we partition the halfspaces based on the signs of the coordinates in their normal vectors.

\begin{definition}[Point-Halfspace and -Halfplane Incidence Graphs]
Let $P$ be a set of points in $\bR^d$ and $H$ be a set of halfspaces in $\bR^d$. 
The \emph{incidence graph} of $P$ and $H$ is the bipartite graph
\[
G(P, H) = \left( P, H, \{ ph ~|~ p \in P, h \in H, \text{ and } p \in h \} \right).
\]
A bipartite graph $G$ is a \emph{point-halfspace incidence} graph in $\bR^d$ if it can be
represented as an incidence graph in $\bR^d$; more specifically, if there exist a set $P$ of points
and a set $H$ of halfspaces both in $\bR^d$ such that $G$ is isomorphic to $G(P, H)$.  If $d=2$ we
call the graph a \emph{point-halfplane incidence} graph.

If there exists such a representation of $G$, where in addition all pairwise dot products of the
norm vectors of the hyperplanes defining the halfspaces in $H$ are non-negative, then $G$ is called
a \emph{positive point-halfspace incidence} graph in $\bR^d$. 
\end{definition}

\begin{proposition}\label{prop:dimension-reduction}
Any bipartite graph $G = (U,W,E)$ of sign-rank $d$ admits a partition $U = U_1 \cup U_2$ such that
$G[U_1,W]$ and $G[U_2,W]$ are point-halfspace incidence graphs in $\bR^{d-1}$.
\end{proposition}
\begin{proof}
	Let $a : U \rightarrow \bR^d$, $b : W \rightarrow \bR^d$ be such that
	$u \in U$, $w \in W$ are adjacent if and only if $\langle a(u), b(w) \rangle \geq 0$.
	We assume, without loss of generality, that $a(u)_d \neq 0$ for every $u \in U$, and partition $U$ into 
	$U_1 = \{ u \in U ~|~ a(u)_d  > 0 \}$ and $U_2 = \{ u \in U ~|~ a(u)_d  < 0 \}$
	
	We define $a' : U \rightarrow \bR^{d-1}$ and $b' : W \rightarrow \bR^{d-1}$ as
	$$a'(u) = \frac{1}{|a(u)_d|}(a(u)_1, a(u)_2, \ldots, a(u)_{d-1}) \text{,~~~~~}
	b'(w) =(b(w)_1, b(w)_2, \ldots, b(w)_{d-1}).$$
	
	Further, we define the following sets of points and halfspaces in $\bR^{d-1}$:
	\begin{align*} 
		P_i &= \{ p_u = a'(u) ~|~ u \in U_i \}, i = 1,2, \\
		H_1 &= \left\{ h_w ~|~ w \in W, h_w = \{ x \in \bR^{d-1} ~|~ \langle x, b'(w) \rangle \geq -b(w)_d \} \right\}, \\
		H_2 &= \left\{ h_w' ~|~ w \in W, h_w' = \{ x \in \bR^{d-1} ~|~ \langle x, b'(w) \rangle \geq b(w)_d \} \right\}.
	\end{align*} 

	Finally, we define two point-halfspace incidence graphs $G_1 = (U_1, W, E_1)$ and $G_2 = (U_2, W, E_2)$, where
	$E_1 = \{ uw ~|~ u \in U_1, w \in W, p_u \in h_w \}$ and $E_2 = \{ uw ~|~ u \in U_2, w \in W, p_u \in h_w' \}$. 
	
	We claim that $G = G_1 \cup G_2 = (U_1 \cup U_2, W, E_1 \cup E_2)$. Indeed, for any $u \in U$ and $w \in W$, 
	\begin{align*} 
		uw \in E 
		\iff
		\langle a(u), b(w) \rangle \geq 0 
		& \iff 
		\langle a'(u), b'(w) \rangle  + \sign(a(u)_d) \cdot b(w)_d \geq 0 \\ 
		& \iff
		\langle a'(u), b'(w) \rangle \geq - \sign(a(u)_d) \cdot b(w)_d.
	\end{align*} 
	Hence, if $u \in U_1$, then $uw \in E \iff \langle a'(u), b'(w) \rangle \geq - b(w)_d \iff p_u \in h_w \iff uw \in E_1$; and
	if $u \in U_2$, then $uw \in E \iff \langle a'(u), b'(w) \rangle \geq b(w)_d \iff p_u \in h_w' \iff uw \in E_2$.
\end{proof}

\begin{proposition}\label{lem:phi-vs-pphi}
Any point-halfspace incidence graph $G = G(P,H)$ in $\bR^{d}$ admits a partition $H =
\bigcup_{i=1}^{2^d} H_i$ such that each $G(P,H_i)$ is a \emph{positive} point-halfspace incidence graph.
\end{proposition}
\begin{proof}
	For $h \in H$, let $w_h \in \bR^d$ and $t_h \in \bR$ be such that $h = \{ x \in \bR^d ~|~  \langle w_h, x \rangle \leq t_h \}$.
	We partition $H$ into $2^d$ subsets $H_{\alpha}$, $\alpha \in \{ -1, +1 \}^d$, with respect to the sign patterns of the norm vectors.
	More specifically, $h \in H_{\alpha}$ if and only if $(w_h)_i \geq 0 \iff \alpha_i = +1$ for every $i \in [d]$.
	Clearly, for any $\alpha \in \{-1,+1\}^d$ and any $h, h' \in H_{\alpha}$, we have that $\langle w_h, w_{h'} \rangle \geq 0$,
	i.e. $G(P, H_{\alpha})$ is a positive point-halfplane incidence graph.
\end{proof}

From \cref{prop:dimension-reduction} and \cref{lem:phi-vs-pphi} we obtain the following immediate

\begin{corollary}
\label{cor:srk-vs-pphi}
Any bipartite graph $G = (U,W,E)$ of sign-rank $d$ admits a partition $U = \bigcup_{i=1}^{2^d} U_i$
such that each $G[U_i,W]$ is a positive point-halfspace incidence graph in $\bR^{d-1}$.
\end{corollary}

We now prove that positive point-halfplane incidence graphs are edge-asteroid-triple free.

\begin{lemma}\label{lem:eat-free}
Every positive point-halfplane incidence graph $G$ is edge-asteroid-triple-free.
\end{lemma}
\begin{proof}
	Let $G$ be a positive point-halfplane incidence graph and let $P$ and $H$ be sets of points and halfplanes respectively
	whose incidence graph is isomorphic to $G$. For a point $p \in P$ we denote by $x_p$ and $y_p$ its
coordinates respectively; for a halfplane $h \in H$ we denote by $a_h, b_h, t_h$ the coefficients of the halfplane inequality, i.e. 
	$h  = \{ (x,y) \in \bR^2 ~|~ a_h x + b_h y \leq t_h \}$.
	Without loss of generality, we can assume that no point in $P$ lies on the boundary of any $h \in
H$.
	Since $G$ is positive, using translation and rotation, we can further assume that for every $h \in
H$ both $a_h$ and $b_h$ are non-negative. This latter assumption implies the following useful claim which is straightforward to verify.
	
	\textbf{Claim 1.} \textit{For every $h \in H$ it holds that if $(x,y) \in h$, then $(x',y') \in h$ for every $x' \leq x$ and $y' \leq y$.}
	%\\ \textit{Proof.} $t_{h} \geq a_{h} x + b_{h} y \geq a_{h} x' + b_{h} y'$. \qed
	
  Suppose now, towards a contradiction, that $G$ contains an edge-asteroid triple, and let $p_1h_1,
p_2h_2,p_3h_3$ be its edges, where $p_i \in P, h_i \in H$. Note that the points $p_1,p_2,p_3$ are
pairwise incomparable with respect to the coordinatewise order. Indeed, if for example $x_{p_1} \leq
x_{p_2}$ and $y_{p_1} \leq y_{p_2}$, then by Claim 1 we would have $p_1 \in h_2$, i.e. $p_1$ and
$h_2$ would be adjacent in $G$, which would contradict the assumption that the three edges form an
edge-asteroid triple.
	
	Thus, without loss of generality, we assume that $x_{p_1} \leq x_{p_2} \leq x_{p_3}$ and $y_{p_1} \geq y_{p_2} \geq y_{p_3}$. 
	
	Let $Q=(q_1,f_1,q_2,f_2,q_3, \ldots, f_{k-1},q_k)$, $q_i \in P$, $f_i \in H$, $q_1 = p_1$ and $q_k = p_3$, be a path containing $p_1$ and $p_3$ that avoids the neighbourhoods of both $p_2$ and $h_2$.
	Since $x_{q_1} \leq x_{p_2} \leq x_{q_k}$, there exists $s \in [k-1]$ such that $x_{q_s} \leq x_{p_2} \leq x_{q_{s+1}}$.
	As above, using Claim 1, we can conclude $q_s$ is incomparable with $p_2$, as otherwise 
	$f_s$ would be adjacent to $p_2$ or $h_2$ would be adjacent to $q_s$, contradicting the choice of $Q$. Similarly, $q_{s+1}$ is incomparable with $p_2$. Hence, we have that $y_{q_s} \geq y_{p_2} \geq y_{q_{s+1}}$.
	
	Let now $\overline{h_2}$ be the closure of the complement of $h_2$, i.e. 
	$\overline{h_2} = \{ (x,y) \in \bR^2 ~|~ a_{h_2} x + b_{h_2} y \geq t_{h_2} \}$, and let
	%Since $p_2$ (and hence all $(x,y) \in \bR^2$ with $x \leq x_{p_2}$ and $y \leq y_{p_2}$) belong to the interior of $h_2$, we can partition $\overline{h_2}$ into three convex sets:
	\begin{itemize}
		\item[] $A_{-+} = \{ (x,y) \in \overline{h_2} ~|~ x \leq x_{p_2}, y \geq y_{p_2} \}$, 
		\item[] $A_{+-} = \{ (x,y) \in \overline{h_2} ~|~ x \geq x_{p_2}, y \leq y_{p_2} \}$.
		%\item[] $A_{++} = \{ (x,y) \in \overline{h_2} ~|~ x > x_{p_2}, y > y_{p_2} \}$.
	\end{itemize}

	Note that $q_s \in A_{-+}$ and $q_{s+1} \in A_{+-}$. Hence the segment connecting $q_s$ and $q_{s+1}$ intersects the line $x = x_{p_2}$ that separates the two sets. Let $(x_{p_2}, y^*) \in \bR^2$ be the point of intersection.
	Since both $q_s$ and $q_{s+1}$ are in $\overline{h_2}$, so is $(x_{p_2}, y^*)$, which together with Claim 1 implies that $y_{p_2} \leq y^*$.
	Similarly, $(x_{p_2}, y^*)$ is in $f_s$ because both $q_s$ and $q_{s+1}$ are in $f_s$.
	Consequently, by Claim 1, $p_2$ is also contained in $f_s$. This contradiction completes the proof.
\end{proof}

\begin{proposition}\label{prop:positive-comp}
	Let $G(P,H)$ be a positive point-halfplane incidence graph. 
	Then the bipartite complement of $G(P,H)$ is also a positive point-halfplane incidence graph.
\end{proposition}
\begin{proof}
	Without loss of generality, we can assume that no point in $P$ lies on the boundary of any $h \in
	H$.
	For every point $p = (x_p,y_p) \in P$ we define $p' := (-x_p,-y_p)$, and for every 
	$h = \{ (x,y) \in \bR^2 ~|~ a_h x + b_h y < t_h \} \in H$ we define $h' = \{ (x,y) \in \bR^2 ~|~ a_h x + b_h y < -t_h \}$.
 	Let $P' = \{ p' ~|~ p \in P \}$ and $H' = \{ h' ~|~ h \in H \}$.
 	We claim that $G(P',H')$ is the bipartite complement of $G(P,H)$.
 	Indeed, for any $p \in P$ and $h \in H$ we have that 
 	$p \in h 
 	\iff a_h x_p + b_h y_p < t_h 
 	\iff a_h (-x_p) + b_h (-y_p) > -t_h 
 	\iff p' \not\in h'$. 
 	Finally, notice that the norm vector of the hyperplane defining a halfspace $h \in H$ is the same as the norm vector of the hyperplane defining $h' \in H'$. Hence, $G(P',H')$ is a positive point-halfplane incidence graph.
\end{proof}

\begin{lemma}
\label{lemma:sign-rank-3-decomposition}
Let $G = (X,Y,E)$ be a bipartite graph of sign-rank 3. Then there exists a partition $Y =
\bigcup_{i=1}^{2^3} Y_i$ such that each $G[X,Y_i]$ is both $(S_{3, 3}, \bc{S}_{3, 3})$-free, and $(C_t, \bc{C}_t)$-free for all $t \geq 10$.
\end{lemma}
\begin{proof}
We claim that the partition $Y = \bigcup_{i=1}^{2^3} Y_i$ given by \cref{cor:srk-vs-pphi} is a desired one. 
Indeed, by \cref{cor:srk-vs-pphi} and \cref{prop:positive-comp}, we conclude that
for each $i \in [2^3]$, the graph $G_i \define G[X,Y_i]$ and its bipartite complement are both positive point-halfplane incidence graphs. Hence, the lemma follows from \cref{lem:eat-free} and \cref{obs:edge-asteroid-triple}.
\end{proof}

\begin{restatable}{theorem}{thmintrosignrank}
\label{thm:intro-signrank}
Let $\cG$ be a graph class with sign-rank at most 3. Then $\DEQ(\cG) = O(1)$ if and only
if $\cG$ is stable.
\end{restatable}
%\thmintrosignrank*

\begin{proof}
It suffices to prove that $\DEQ(\cG) = O(1)$ when $\cG$ is stable, due to
\cref{prop:gt-lower-bound}.  On graph $G = (X,Y,E)$ and inputs $x \in X$, $y \in Y$, the players
compute the decomposition $Y = \bigcup_{i=1}^{8} Y_i$ given by
\cref{lemma:sign-rank-3-decomposition} and use 3 bits of communication to agree on the value $i$
such that $y \in Y_i$. Then they compute adjacency in $G[X,Y_i]$ by applying the protocol in
\cref{thm:intro-main}.
\end{proof}

\subsection{Unit Disk Graphs}
\label{section:unit-disk-graphs}

%We begin with the proof of our \cref{thm:intro-udg} for unit disk graphs. 
In this section we prove our result for unit disk graphs.  A graph $G$ is \emph{unit disk} if there
exists  a mapping $\phi : V(G) \to \bR^2$ such that $xy \in E(G)$ if and only if
$\|\phi(x)-\phi(y)\|_2 < 2$.  The mapping $\phi$ is called a \emph{realisation} of $G$. Note that
the constant 2 may be replaced with any other constant.  We start by observing that unit disk graphs
have sign-rank at most 4.  

\begin{fact}\label{cl:UDG-srk-4}
Any unit disk graph $G$ has sign-rank at most 4.
\end{fact}
\begin{proof}
	Let $v \mapsto (x_v, y_v) \in \bR^2$ for $v \in V(G)$, be a realisation of $G$, such that for any two
distinct vertices $a,b \in V$, $ab \in E(G)$ if and only if
\[
	(x_a - x_b)^2 + (y_a - y_b)^2 < \sqrt{2}
	\iff
	x_a^2 - 2x_a x_b + x_b^2 + y_a^2 - 2y_a y_b + y_b^2 - \sqrt{2} < 0
\]
Then, by defining 
$\sigma : v \mapsto (-1, 2x_v, 2y_v, -x_v^2 - y_v^2) \in \bR^4$,
$\psi : v \mapsto (x_v^2 + y_v^2 -\sqrt{2}, x_v, y_v, 1) \in \bR^4$
for $v \in V(G)$, we see that for any \emph{distinct} $a,b \in V(G)$
\[
  \inn{\sigma(a), \psi(b)} > 0 \iff (x_a - x_b)^2 + (y_a - y_b)^2 < \sqrt{2} \,,
\]
so $\sign(\inn{\sigma(a), \psi(b)}) = 1$ if and only if $ab \in E$, as desired.
\end{proof}

The main tool for our application to unit disk graphs is the following lemma of \cite{AZ18}.  A
graph $G$ is \emph{co-bipartite} if its complement $\overline G$ is bipartite. 

\begin{lemma}[\cite{AZ18}]
\label{lemma:az18}
Let $G$ be a co-bipartite unit disk graph. Then the bipartite graph
$\overline G$ and its bipartite complement do not contain any edge-asteroid triples. In particular,
due to \cref{obs:edge-asteroid-triple}, $\overline G$ is both 
$(S_{3, 3}, \bc{S}_{3, 3})$-free, and $(C_t, \bc{C_t})$-free for all $t \geq 10$.
\end{lemma}

Our upper bound on the communication complexity of stable unit disk graphs will follow from a fairly
straightforward decomposition of a unit disk graph into unit-length grid cells, such that between
any two grid cells the graph is co-bipartite.

\begin{restatable}{theorem}{thmintroudg}
\label{thm:intro-udg}
Let $\cG$ be a subclass of unit disk graphs. Then $\DEQ(\cG) = O(1)$ if and only if $\cG$ is stable.
\end{restatable}
%\thmintroudg*
\begin{proof}
It suffices to show that if $\cG$ is stable, then $\DEQ(\cG) = O(1)$, due to
\cref{prop:gt-lower-bound}.  Since $\cG$ is stable, there exists a constant $k$ such that for all $G
\in \cG$, $\ch(G) < k$. Fix any $G \in \cG$ together with its realisation $\phi : V(G) \to \bR^2$. 
For convenience, we will identify the vertices $x \in V(G)$ with the corresponding points $\phi(x) \in
\bR^2$. On inputs $x,y \in V(G)$, Alice and Bob will perform the following protocol.
\begin{enumerate}
\item Alice and Bob each partition $\bR^2$ into a grid with cells $C_{i,j}$ for $i,j \in \bZ$, where
$C_{i,j} \define \{ (z_1,z_2) \in \bR^2 : i \leq z_1 < i+1, j \leq z_2 < j+1 \}$. Observe that if $x,y$ are
adjacent, then if $x \in C_{i,j}$ we must have $y \in C_{i+a, j+b}$ for some $a,b \in \{-2,-1,0,1,2\}$;
and if $x,y \in C_{i,j}$ then $\|x-y\|_2 < \sqrt 2$ so $x,y$ are adjacent.  Let $i_x, j_x \in \bZ$
be such that $x \in C_{i_x, j_x}$ and let $i_y, j_y \in \bZ$ be such that $y \in C_{i_y, j_y}$.
\label{item:protocol-udg-1}
\item Using 1 call to the $\Eq$ oracle, Alice and Bob check if $(i_x, j_x) = (i_y,j_y)$. If so, they
output 1 and the protocol terminates. In this case, the protocol is correct due to the observation
above.
\item For each $(a,b), (a',b') \in \{-2,-1,0,1,2\}^2$ such that $(a,b) \neq (0,0)$ and $(a',b') \neq
(0,0)$, Alice and Bob use 2 calls to the $\Eq$ oracle to check if 
both $(i_x+a,j_x+b) = (i_y,j_y)$. If so, then Alice and Bob compute adjacency in the semi-induced
bipartite graph $G[X,Y]$ where $X \define C_{i_x,j_x} \cap V(G)$ and $Y \define C_{i_y,j_y} \cap
V(G)$. This is possible because:
\begin{itemize}
\item Alice and Bob each know $X$ and $Y$: Alice knows $(i_x+a,j_x+b) = (i_y,j_y)$ and Bob knows
$(i_y+a',j_y+b')=(i_x,j_x)$; and
\item The graph $G[X,Y]$ has $\ch(G[X,Y]) \leq \ch(G) \leq k$, and it is
$(S_{3, 3}, \bc{S}_{3, 3})$-free, and $(C_t, \bc{C_t})$-free for all $t \geq 10$, by \cref{lemma:az18}, so we may apply \cref{thm:intro-main}. 
\end{itemize}
\item If $(i_x + a,j_x + b) \neq (i_y,j_y)$ for all $(a,b) \in \{-2,-1,0,1,2\}^2$, then Alice and Bob
output 0. In this case the protocol is correct by the observation in step \ref{item:protocol-udg-1}.
\end{enumerate}
This concludes the proof.
\end{proof}

\section{The Sign-Rank Hierarchy}
\label{section:hierarchy}

We have now determined exactly the conditions required for graphs of sign-rank 3, and some graphs of
sign-rank 4, to have constant randomized communication cost (equivalently, constant margin). Let us
now consider sign-ranks 5 and above. We will see that our techniques for the lower sign-ranks,
specifically the reduction to \textsc{Equality}, will surely fail.  This witnesses a certain
threshold between sign-ranks 3 and 5.

It is common to study the bipartite graphs which are $K_{t,t}$-free, for some constant $t$. If a
class of graphs is $K_{t,t}$-free it is called \emph{weakly-sparse}. This is a much stronger
condition than stability: if $G$ is $K_{t,t}$-free then it must satisfy $\ch(G) \leq 2t$.

A hereditary graph class $\cG$ has \emph{bounded degeneracy} (equivalently, \emph{bounded
arboricity}) if there exists a constant $d$ such that every $G \in \cG$ has a vertex of degree at
most $d$ (equivalently, if there exists a constant $a$ such that every $G \in \cG$ on $N$ vertices
has at most $a \cdot N$ edges). The following is well-known and easy to prove (see \eg
\cite{Har20,HWZ22}).

\begin{proposition}
If a hereditary graph class $\cG$ has bounded degeneracy then $\DEQ(\cG) = O(1)$.
\end{proposition}

Recent results of \cite{CH23} show that, for any constant $t$, the $K_{t,t}$-free point-halfspace
incidence graphs in dimension 3 have bounded degeneracy. Since any graph of sign-rank 4 can be
written as a union of two point-halfspace incidence graphs in dimension 3 (see
\cref{prop:dimension-reduction}), we obtain the following:

\begin{theorem}
Let $\cG$ be a hereditary graph class that is weakly-sparse and has sign-rank at most 4. Then
$\DEQ(\cG) = O(1)$.
\end{theorem}

Our \cref{thm:intro-signrank} and \cref{thm:intro-udg} are strengthenings of this theorem for the
special cases of sign-rank 3 and unit disk graphs, where we replace the \emph{weakly-sparse}
condition with the much less restrictive stability condition.  The proof of \cite{CH23} uses a
technique based on ``shallow cuttings''\!. In \cref{section:caratheodory}, we give a simpler proof,
using elementary geometry, of the weaker statement that $K_{2,t}$-free point-halfspace incidence
graphs in dimension 3 have bounded degeneracy.

However, it is not possible to extend these results even to sign-rank 5. To see this, we first
require a theorem which shows that, for weakly-sparse classes, bounded degeneracy is equivalent to
the existence of a reduction to \textsc{Equality}.
\iftoggle{anonymous}{%ANONYMOUS
%We defer the proof of this theorem to \cref{section:bonamy-girao-proof}.
}{%NON-ANONYMOUS
The proof generalizes a theorem of \cite{HWZ22} and is due to Bonamy, Esperet, \& Gir{\~ao}, which
we include here with permission and gratitude.
}

\thmeqlowerbound*

This theorem follows from the next two lemmas. The first lemma is implicit in \cite{HWZ22}.

\begin{lemma}
Let $\cG$ be a class of bipartite graphs satisfying the following Ramsey property: for any $k, \ell
\in \bN$ there exists a graph $G \in \cG$ such that, for any coloring of the edges of $G$ with at
most $k$ colors, there exists a monochromatic induced path on $\ell$ vertices. Then $\D^\Eq(\cG) =
\omega(1)$.
\end{lemma}

This lemma was used in \cite{HWZ22} in conjunction with a result of \cite{ARSV06} that established
the required Ramsey property for induced subgraphs of hypercubes, which are $K_{2,3}$-free, to show
that hypercubes do not have constant-cost reductions to \textsc{Equality}.  The next lemma
generalizes this result.

\iftoggle{anonymous}{% ANONYMOUS
\begin{lemma}
\label{lemma:bonamy-girao}
For any $k, t, \ell \in \bN$, there is an integer $d$ such that, if a $K_{t,t}$-free bipartite graph
$G$ has average degree at least $d$, and its edges are colored with at most $k$ colors, then $G$
contains a monochromatic induced path of length at least $\ell$.
\end{lemma}
}{% NON-ANONYMOUS
\begin{lemma}[Bonamy, Esperet, \& Gir{\~a}o]
\label{lemma:bonamy-girao}
For any $k, t, \ell \in \bN$, there is an integer $d$ such that, if a $K_{t,t}$-free bipartite graph
$G$ has average degree at least $d$, and its edges are colored with at most $k$ colors, then $G$
contains a monochromatic induced path of length at least $\ell$.
\end{lemma}
}
\begin{proof}
We first reduce to the case $t=2$. Suppose $t > 2$ and let $d \geq t$. A result of \cite{McC21}
shows that there is a constant $d'$ such that any $K_{t,t}$-free bipartite graph of average degree
at least $d'$ contains a $K_{2,2}$-free induced subgraph of average degree at least $d$. Therefore
it suffices to consider the case $t=2$.

Choose $b > \ell$ and set $d = 2kb$.  Consider a $K_{2,2}$-free graph $G$, whose edges are colored
with at most $k$ colors. Then, if the average degree of $G$ is at least $d$, there exists a color $c
\in [k]$ such that the graph $G_c$ induced by the edges with color $c$ has average degree at least
$2b$. Then $G_c$ has an induced subgraph $G'_c$ with \emph{minimum} degree at least $b$.

We now construct a monochromatic induced path in $G$ by induction as follows. The base case, a
monochromatic path on 2 vertices, is trivial. Suppose we have obtained an induced path $P_{s-1} = \{
v_1, \dotsc, v_{s-1} \}$, for $s-1 < \ell$ where each $(v_i, v_{i+1})$ is an edge of $G'_c$. Let
$N'_c(v_{s-1})$ be the neighbors of $v_{s-1}$ in $G'_c$ and suppose for contradiction that all
vertices $u \in N'_c(v_{s-1}) \setminus P_{s-1}$ are adjacent in $G$ to some $v_i$ with $i < s-1$.
Since $v_{s-1}$ has at least $b > \ell > s-1$ neighbors, there are two vertices $u,w \in N'_c(v_{s-1})
\setminus P_{s-1}$ that are adjacent in $G$ to both $v_{s-1}$ and $v_i$ for some $i < s-1$. But then
$\{v_i, v_{s-1}, u, w\}$ form an induced $K_{2,2}$, which is a contradiction. Therefore there exists
a vertex $v_s \in N'_c(v_{s-1}) \setminus P_{s-1}$ which produces a monochromatic induced path $P_s
= \{ v_1, \dotsc, v_s \}$. This concludes the proof.
\end{proof}

\subsection{Sign-Rank 5: Point-Box Incidences}

To show that our techniques cannot extend to sign-rank 5, even if we ask for the much stronger
$K_{2,2}$-free condition instead of stability, it now suffices to show that there exists a
weakly-sparse class of bipartite graphs with sign-rank 5 and unbounded degeneracy. For this we use
the \emph{point-box incidence} graphs.

\begin{definition}[Point-Box Incidences]
Let $P$ be a set of points in $\bR^2$ and $H$ a set of axis-aligned rectangles in $\bR^2$. The
\emph{incidence graph} of $P$ and $H$ is the bipartite graph
\[
  G(P, H) \define (P, H, \{ ph | p \in P, h \in H, p \in h \}) \,.
\]
\end{definition}

The fact that these graphs have sign-rank 5 follows from a transformation of point-box incidences in
dimension $2$ to point-halfspace incidences in dimesion $4$, which appears in \cite{PT13,CMK19}. The sign-rank
of point-halfspace incidences in $\bR^4$ is at most 5.

\begin{lemma}[\cite{CMK19}]
\label{lemma:point-box-sign-rank}
The class of point-box incidence graphs has sign-rank at most 5.
\end{lemma}

What remains is the claim that \emph{weakly-sparse} point-box incidence graphs on $N$ vertices can have
$\omega(N)$ edges.  This is true even under the strongest condition of being $K_{2,2}$-free.  The
lower bound of the next lemma was proved recently in \cite{BCSTT21}, and the upper bound in
\cite{CH23}.  We remark that the lemma remains true even if the boxes are restricted to be
\emph{dyadic}, \ie the product of intervals of the form $[s2^t, (s+1)2^t)$ with integers $s,t$.

\begin{lemma}[\cite{BCSTT21,CH23}]
\label{lemma:point-box-edges}
The maximum number of edges in a $K_{2,2}$-free point-box incidence graph is $\Theta\left( n \cdot
\frac{\log n}{\log\log n}\right)$. As a consequence, $K_{2,2}$-free point-box incidence graphs have
unbounded degeneracy.
\end{lemma}

Combining \cref{thm:eq-lower-bound,lemma:point-box-sign-rank,lemma:point-box-edges}, we get:

\begin{corollary}
There is a hereditary class $\cG$ of $K_{2,2}$-free bipartite graphs with sign-rank 5 and $\DEQ(\cG)
= \omega(1)$.
\end{corollary}

\subsection{Sign-Rank 6: Point-Line Incidences}

The above result shows that reductions to \textsc{Equality} cannot be used to prove $\R(\cG) = O(1)$
in general, even for weakly-sparse classes, let alone stable ones. This leaves open the possibility
that there is another method for obtaining constant-cost randomized communication protocols for
weakly-sparse or even stable graph classes with sign-rank 5, 6, or any constant. However, we discuss
here a recent conjecture of \cite{CHHS23} regarding point-line incidences suggesting that weakly-sparse graphs of sign-rank 6 have non-constant communication complexity.

\begin{definition}[Point-Line Incidences]
Let $P$ be a set of points in $\bR^2$ and $L$ be a set of lines in $\bR^2$. The \emph{incidence
graph} of $P$ and $L$ is the bipartite graph
\[
  G(P, L) \define (P, L, \{ ph | p \in P, \ell \in L, p \in \ell \}) \,.
\]
\end{definition}

Point-line incidence graphs are $K_{2,2}$-free by definition, and it is well-known that the
incidence graph between $N$ points and $N$ lines can have $\Theta(N^{4/3})$ edges; therefore,
\cref{thm:eq-lower-bound} guarantees that they do not reduce to \textsc{Equality}. 
Furthermore, it is known that point-line incidence graphs are point-halfspace incidence graphs in $\bR^5$ (see e.g. \cite{CH23}), and hence they have sign-rank at most 6:

\begin{proposition}
Point-line incidence graphs have sign-rank at most 6.
\end{proposition}

The communication complexity of point-line incidence graphs was recently studied in \cite{CHHS23},
but it remains unknown whether they have constant-cost. It was conjectured that they do not:

\begin{conjecture}[\cite{CHHS23}]
\label{conj:point-line}
The class $\cG$ of point-line incidence graphs has $\R(\cG) = \omega(1)$.
\end{conjecture}

\iftoggle{anonymous}{%ANONYMOUS
}{% NON-ANONYMOUS
\begin{center}
\Large\textbf{Acknowledgments}
\end{center}
We are grateful to Marthe Bonamy, Louis Esperet, and Antonio Gir{\~a}o for their proof of
\cref{lemma:bonamy-girao}, and to Louis Esperet for communicating this proof to us and allowing us
to include it here. We thank Lianna Hambardzumyan, Pooya Hatami, and Sebastian Wild for several
conversations on the topic of this paper. The general definition of constant-cost reductions given
in this paper has arisen partly out of collaboration with Yuting Fang, Lianna Hambardzumyan, and
Pooya Hatami. We thank M{\'o}nika Csik{\'o}s for telling us about \cref{lemma:point-box-sign-rank}.
}

\bibliographystyle{alpha}
\bibliography{references}

\appendix

%%%%%%%%%%%%%%%%%%%%%%%%%%%%%%%%%%%%%%
\section{On the number of edges in weakly-sparse point-halfspace incidence graphs}
\label{section:caratheodory}
%%%%%%%%%%%%%%%%%%%%%%%%%%%%%%%%%%%%%%

In this section we show that $K_{2,s}$-free point-halfspace incidence graphs in 
dimensions 1,2, and 3 have linearly many edges. 
The same result was recently obtained by Chan and Har-Peled in \cite{CH23} for more
general classes of $K_{s,s}$-free graphs. 
We present our results for two reasons. First, the proof technique is completely different and might be of independent interest. Second, our bounds are more specific for the considered cases.

To prove our upper bounds, we will show that every graph in a class has a vertex of bounded degree.
Since the classes are hereditary, this will imply linear bounds on the number of edges.

\subsection{On the line}

In this section we will show that the $K_{s,s}$-free point-halfline incidence graphs on $\bR$ have linear number of edges.
In fact we will show a linear bound on the number of edges in the more general class of the $K_{s,s}$-free point-interval incidence graphs. 
For this latter class, \cite{CH23} shows that an $n$-vertex $K_{s,s}$-free point-interval incidence graphs with $n_p$ points and $n_i$ intervals contains at most $s(n_p+3n_i)$ edges. 
Our bound of $(s-1)n = (s-1)(n_p+n_i)$ is a slight improvement over the bound from \cite{CH23}.

\begin{lemma}\label{lem:R1}
	Let $G$ be a $K_{s,s}$-free $n$-vertex point-interval incidence graph. Then $G$ has at most $(s-1)n$ edges.
\end{lemma}
\begin{proof}
	Let $P$ be a set of points and $I$ be a set of intervals on the real line such that 
	$G \simeq G(P, I)$.
	To prove the statement we will show that $G$ has a vertex of degree at most $s-1$. Suppose that all vertices of $G$ have degree at least $s$ and let $p$ be the leftmost point in $P$. The degree assumption implies that $p$ belongs
	to at least $s$ intervals, which we denote $i_1, i_2, \ldots, i_s$. For the same reason, each of these intervals should contain the $s-1$ points in $P$ closest to $p$, which we denote $p_1, p_2, \ldots, p_{s-1}$.
	But then the vertices corresponding to $i_1, i_2, \ldots, i_s$ and $p, p_1, p_2, \ldots, p_{s-1}$ induce
	the forbidden $K_{s,s}$.
\end{proof}

%%%%%%%%%%%%%%%%%%%%%%%%%%%%%%%%%%%%%%
\subsection{On the plane}
\label{section:PH-on-the-plane}
%%%%%%%%%%%%%%%%%%%%%%%%%%%%%%%%%%%%%%

In dimensions 2 and 3, the bounds for $K_{s,s}$-free graphs from \cite{CH23} are $O(sn)$, and the constants in the big-$O$ are not specified.

Our bounds in dimension 2 and 3 are respectively $3(s-1)n$ and $5(s-1)n$.  
To obtain them we will use the following lemma that reduces the
analysis to the case where the points are in convex position.

\begin{lemma}\label{lem:Caratheodory}
	Let $G \simeq G(P, H)$ be the incidence graph of a set $P$ of points and a set $H$ of
	halfspaces in $\bR^d$. If $G$ is $K_{2,s}$-free and $P$ is not in convex position, then $G$
	has a vertex of degree at most $(d+1)(s-1)$.
\end{lemma}
\begin{proof}
	Suppose that $P$ is not in convex position, and let $p \in P$ be a non-extremal point of the
	convex hull $\conv(P)$. 
	By Carath{\'e}odory's theorem, $p$ belongs to the convex hull of at most $d+1$
	extremal points of $\conv(P)$. Let $p_1,p_2, \ldots, p_k$, $k \leq d+1$ be a minimal set of such extremal points.
	Since $p$ belongs to the interior of $\conv(\{p_1, \ldots, p_k\})$, any halfspace containing $p$ contains one of the points 
	$p_1, \ldots, p_k$. Thus, if $p$ belongs to at least $k(s-1)+1$ halfspaces, one of the points  $p_1, \ldots, p_k$ belongs to at least $s$ of them resulting in the forbidden $K_{2,s}$.
	Hence, the degree of $p$ is at most $k(s-1) \leq (d+1)(s-1)$.
\end{proof}

%\begin{lemma}\label{lem:R2-covex}
%	Let $G \simeq G(P, H)$ be the incidence graph of the set $P$ of points and the set $H$ of
%	halfspaces in $\bR^2$. If $G$ is $K_{s,s}$-free and $P$ is in convex position, then $G$
%	has a vertex of degree at most $2(s-1)$.
%\end{lemma}
%\begin{proof}
%	Suppose that $P$ is in convex position, i.e. all points in $P$ are extremal points of $\conv(P)$,
%	and in particular 1-skeleton of $\conv(P)$ is a cycle on $|P|$ vertices. 
%	We will denote this cycle as $C$.
%	
%	Suppose that all vertices of $G$ have degree at least $2(s-1)+1$, and let $p \in P$ be an arbitrary
%	vertex of $G$. By the degree assumption, $p$ belongs to at least $2(s-1)+1$ halfspaces from $\cH$.
%	For the same reason, each of the halfspaces containing $p$ contains either the $s-1$ points
%	that are on the counter-clockwise path of length $s-1$ on $C$ starting from $p$ or
%	the $s-1$ points that are on the clockwise path of length $s-1$ on $C$ starting from $p$.
%	Thus, by the pigeonhole principle, there are at least $s$ of the halfspaces that contain $p$
%	and the same set of $s-1$ more points. These halfspaces and the points induce the forbidden $K_{s,s}$ in $G$. Hence, $G$ should contain a vertex of degree at most $2(s-1)$.
%\end{proof}
%
%Now the bound on the number of edges is $K_{2,s}$-free $n$-vertex point-halfplane incidence 
%graphs follow from \cref{lem:Caratheodory,lem:R2-covex}.

The \emph{polytope graph} of a polytope is the incidence graph
of the extremal points and 1-dimensional faces of the polytope. We will need the following
well-known fact.

\begin{fact}\label{fact:cut-poltope}
	Let $P$ and $H$ be respectively a polytope and a halfspace in $\bR^d$.
	The subgraph of the polytope graph of $P$ induced by the extremal point of $P$ that belong to $H$ is connected.
\end{fact}

\begin{lemma}\label{lem:R2}
	Let $G$ be a $K_{2,s}$-free $n$-vertex point-halfplane incidence graph.
	Then $G$ has at most $3(s-1)n$ edges.
\end{lemma}
\begin{proof}
		Let $P$ be a set of points on the plane and $H$ be a set of halfplanes such that 
		$G \simeq G(P, H)$. We assume without loss of generality that $|P| \geq 3$.
		To prove the lemma we will show that $G$ has a vertex of degree at most $3(s-1)$.
		If $P$ is not in convex position, such a vertex exists by \cref{lem:Caratheodory}, so we can assume that all points in $P$ are extremal points of $\conv(P)$. 
		Suppose that all vertices of $G$ have degree at least $3(s-1)+1$ and let $p$ be an arbitrary
		point in $P$. 
		The polytope graph of $P$ is a cycle, and hence $p$ has exactly 2 neighbours in this graph.  \cref{fact:cut-poltope}  implies that each of the halfplanes that contain $p$ and some other vertices in $P$ should also contain at least one of these 2 neighbours.
		Thus, since $p$ belongs to $3(s-1)+1 \geq 2(s-1)+1$ halfplanes in $H$, at least
		$s$ of them contain one other fixed point in $P$, which witnesses a forbidden $K_{2,s}$.
\end{proof}

%%%%%%%%%%%%%%%%%%%%%%%%%%%%%%%%%%%%%%
\subsection{In $\bR^3$}
\label{section:PH-in-R3}
%%%%%%%%%%%%%%%%%%%%%%%%%%%%%%%%%%%%%%

\begin{lemma}\label{lem:R3}
	Let $G$ be a $K_{2,s}$-free $n$-vertex point-halfspace incidence graph in $\bR^3$.
	Then $G$ has at most $5(s-1)n$ edges.
\end{lemma}
\begin{proof}
	Let $P$ and $H$ be respectively a set of points and a set of halfspaces in $\bR^3$ such that 
	$G \simeq G(P, H)$. As before, to prove the lemma we will show that $G$ has a vertex of degree at most $5(s-1)$. Towards a contraction, suppose that all vertices in $G$ have at least $5(s-1)+1$ neighbours. This assumption and \cref{lem:Caratheodory} imply that $P$ is in convex position, and hence all points in $P$ are extremal points of $\conv(P)$. 
	
	Let $F$ be the polytope graph of $P$. By Steinitz's theorem (see e.g. \cite{Ziegler95}), $F$ is planar, and therefore has a vertex of degree at most 5. Let $p \in P$ be such a vertex.
	It follows from \cref{fact:cut-poltope} that any halfspace in $H$ that contains $p$ also contains
	at least one of the neighbours of $p$ in $F$. Thus, by the pigeonhole principle, at least $s$ halfspaces among those in $H$ that contain $p$ contain also one fixed neighbour of $p$, which
	witnesses the forbidden $K_{2,s}$. This contradiction completes the proof.
\end{proof}

\end{document}